\documentclass[12pt, draftclsnofoot, onecolumn]{IEEEtran}
\usepackage[T1]{fontenc}

\def\argmax{\mathop{\rm \arg\!\max}}

\usepackage{graphicx}
\usepackage[noadjust]{cite}
\usepackage{mcite}
\usepackage{amsfonts,helvet}
\usepackage{fancyhdr}
\usepackage{threeparttable}
\usepackage{epsf,epsfig}
\usepackage{amsthm}
\usepackage{amsmath}
\usepackage{siunitx}
\usepackage{amssymb}
\usepackage{dsfont}
\usepackage{subfigure}
\usepackage{color}
\usepackage{enumerate}
\usepackage{hyperref}
\usepackage{cancel}
\usepackage{bbm}
\usepackage{dsfont}
\usepackage[subnum]{cases}
\usepackage{adjustbox}
\usepackage[linesnumbered,ruled,noend]{algorithm2e}
\usepackage{multicol}
\usepackage[english]{babel}
\usepackage{enumitem}

\newtheorem{theorem}{Theorem}

\newtheorem{corollary}{Corollary}

\newtheorem{definition}{Definition}

\newtheorem{lemma}{Lemma}

\newtheorem{remark}{Remark}

\newmuskip\pFqmuskip

\def\bn{{\bf n}}

\def\bp{{\bf p}}
\def\bq{{\bf q}}
\def\br{{\bf r}}
\def\bs{{\bf s}}

\def\bu{{\bf u}}
\def\bv{{\bf v}}

\def\bx{{\bf x}}
\def\by{{\bf y}}
\def\bz{{\bf z}}

\def\bA{{\bf A}}
\def\bB{{\bf B}}

\def\bD{{\bf D}}

\def\bG{{\bf G}}
\def\bH{{\bf H}}
\def\bI{{\bf I}}

\def\bK{{\bf K}}

\def\bP{{\bf P}}
\def\bQ{{\bf Q}}
\def\bR{{\bf R}}

\def\bW{{\bf W}}

\def\cA{\mbox{$\mathcal{A}$}}
\def\cB{\mbox{$\mathcal{B}$}}
\def\cC{\mbox{$\mathcal{C}$}}

\def\cK{\mbox{$\mathcal{K}$}}

\def\cN{\mbox{$\mathcal{N}$}}
\def\cO{\mbox{$\mathcal{O}$}}
\def\cP{\mbox{$\mathcal{P}$}}

\def\cR{\mbox{$\mathcal{R}$}}
\def\cS{\mbox{$\mathcal{S}$}}
\def\cT{\mbox{$\mathcal{T}$}}

\def\cV{\mbox{$\mathcal{V}$}}

\def\bbC{\mbox{$\mathbb{C}$}}

\def\bbE{\mbox{$\mathbb{E}$}}

\def\bbR{\mbox{$\mathbb{R}$}}

\def\bbV{\mbox{$\mathbb{V}$}}

\def\ubG{\mbox{$\underline{\bf{G}}$}}
\def\ubH{\mbox{$\underline{\bf{H}}$}}

\def\ubW{\mbox{$\underline{\bf{W}}$}}

\def\ubn{\mbox{$\underline{\bf{n}}$}}

\def\ubq{\mbox{$\underline{\bf{q}}$}}
\def\ubr{\mbox{$\underline{\bf{r}}$}}
\def\ubs{\mbox{$\underline{\bf{s}}$}}

\def\ubu{\mbox{$\underline{\bf{u}}$}}
\def\ubv{\mbox{$\underline{\bf{v}}$}}

\def\ubx{\mbox{$\underline{\bf{x}}$}}
\def\uby{\mbox{$\underline{\bf{y}}$}}
\def\ubz{\mbox{$\underline{\bf{z}}$}}

\usepackage{array}

\makeatletter
\newcommand{\thickhline}{%
    \noalign {\ifnum 0=`}\fi \hrule height 1pt
    \futurelet \reserved@a \@xhline
}
\newcolumntype{"}{@{\hskip\tabcolsep\vrule width 1pt\hskip\tabcolsep}}
\makeatother

\IEEEoverridecommandlockouts

\title{Base Station Antenna Selection for Low-Resolution ADC Systems}

\author{
Jinseok Choi, Junmo Sung, Narayan Prasad, Xiao-Feng Qi, \\Brian L. Evans, and Alan Gatherer \footnote{
J. Choi, J. Sung, and B. L. Evans conducted this research in the Wireless Networking and Communications Group (WNCG), Department of Electrical and Computer Engineering, The University of Texas at Austin, Austin, TX, USA. (e-mail: \{jinseokchoi89,  junmo.sung\}@utexas.edu, bevans@ece.utexas.edu).
J. Sung is also currently with Samsung in Dallas, TX, USA.
N. Prasad and X. Qi are with the Radio Algorithms Research Group, Futurewei Technologies, NJ Research Center, Bridgewater, NJ, USA. (email: \{narayan.prasad, xiao.feng.qi\}@futurewei.com).
A. Gatherer is with Wireless Access Lab., Futurewei Technologies, Legacy Dr, Plano, TX,  USA. (e-mail: alan.gatherer@futurewei.com).
The authors at The University of Texas at Austin were supported by WNCG Industrial Affiliates Funding from Futurewei Technologies received in 2018.
A preliminary version of this work was presented at IEEE ICASSP 2018 \cite{choi2018antenna}.
}
}
\begin{document}
\maketitle

\begin{abstract}
This paper investigates antenna selection at a base station with large antenna arrays and low-resolution analog-to-digital converters. For downlink transmit antenna selection for narrowband channels, we show (1) a selection criterion that maximizes sum rate with zero-forcing precoding equivalent to that of a perfect quantization system; (2) maximum sum rate increases with number of selected antennas; (3) derivation of the sum rate loss function from using a subset of antennas; and (4) unlike high-resolution converter systems, sum rate loss reaches a maximum at a point of total transmit power and decreases beyond that point to converge to zero. For wideband orthogonal-frequency-division-multiplexing (OFDM) systems, our results hold when entire subcarriers share a common subset of antennas. For uplink receive antenna selection for narrowband channels, we (1) generalize a greedy antenna selection criterion to capture tradeoffs between channel gain and quantization error; (2) propose a quantization-aware fast antenna selection algorithm using the criterion; and (3) derive a lower bound on sum rate achieved by the proposed algorithm based on submodular functions. For wideband OFDM systems, we extend our algorithm and derive a lower bound on its sum rate. Simulation results validate theoretical analyses and show increases in sum rate over conventional algorithms.

\end{abstract}

\section{Introduction}
\label{sec:intro}

Large-scale multiple-input multiple-output (MIMO) systems have been considered as a potential technology for future wireless systems because they offer orders of magnitude improvement in spectral efficiency \cite{marzetta2010noncooperative, larsson2014massive}.
The large number of antennas, however, comes with practical challenges such as hardware cost and power consumption \cite{lu2014overview}.
Antenna selection can be a potential solution to reduce the large power consumption by efficiently reducing the number of radio frequency (RF) chains \cite{mendez2016hybrid}.
In addition, since the power consumption of analog-to-digital converters (ADCs) scales exponentially in the number of quantization bits \cite{walden1999analog}, reducing the resolution of ADCs provides additional power savings for future communication systems \cite{mezghani2007ultra, mo2015capacity}.
In this regard, we investigate base station (BS) antenna selection problems in low-resolution ADC systems for uplink (UL) and downlink (DL) communications.


\subsection{Prior Work}

Antenna selection problems have been widely studied without quantization error for high-resolution ADC systems.
For the transmit antenna selection, it was shown that single antenna selection achieves full diversity gain which the transmitter without antenna selection (the transmitter uses all antennas) achieves \cite{chen2005analysis}, and it is optimal in the low signal-to-noise ratio (SNR) \cite{sanayei2007capacity}.
To find the best transmit antenna subset, convex optimization techniques were adopted by relaxing a binary integer problem to a real number problem \cite{gao2013antenna,khademi2014convex}.
Transmit antenna selection was also jointly studied with other problems \cite{zhang2008performance, amadori2017large}.
An outage probability was derived for single user selection and antenna selection in \cite{zhang2008performance}, and a precoder was designed jointly with antenna selection \cite{amadori2017large}.
Energy and spectral efficiency tradeoff was maximized in \cite{liu2017energy} by solving a multi-objective antenna selection problem.
For special systems such as spatial modulation systems, a Euclidian distance-based antenna selection method was developed \cite{yang2016transmit}. 

Receive antenna selection methods were also developed for last decade \cite{gorokhov2003receive, molisch2005cap, dua2006receive, vaze2012submodularity, liu2009low, nara2009antenna}.
In \cite{gorokhov2003receive}, a greedy antenna selection method was developed by minimizing capacity loss.
It was shown in \cite{gorokhov2003receive} that the diversity order of the receive antennal selection system is same as the full diversity order.
In \cite{molisch2005cap}, a correlation-based method and mutual information-based method were developed, showing that selecting receive antennas more than the number of transmit antennas can nearly achieve the performance of full receive antenna systems.
Convex optimization approach was also taken in receive antenna selection \cite{dua2006receive}. 
To provide a lower bound of greedy selection methods, modularity and submodularity concepts were used in \cite{vaze2012submodularity}.
In \cite{zhang2009receive} a sampling-based selection method was proposed by employing cross entropy optimization technique.

Antenna selection problems have been studied for various channels.
For correlated channels, selection algorithms were proposed by exploiting partial channel state information (CSI) such as a channel covariance matrix \cite{dai2006optimal}.
Antenna selection problems were also solved for millimeter wave channels jointly with precoder design \cite{amadori2015low,li2019joint}.
In orthogonal frequency division multiplexing (OFDM) systems, both transmit antenna selection \cite{torabi2008antenna, le2016antenna} and 
receive antenna selection algorithms \cite{liu2009low, nara2009antenna} were developed.
An adaptive Markov chain Monte Carlo (MCMC) method was adopted for antenna selection \cite{liu2009low}, and optimal power allocation between training and data symbols with antenna selection was derived to minimize performance loss due to channel estimation error \cite{nara2009antenna}.
An outage probability was analyzed for per-subcarrier antenna selection in \cite{torabi2008antenna}, and an adaptive antenna selection method that balances between per-subcarrier and bulk selection was proposed in \cite{le2016antenna}.


Most prior work on antenna selection, however, focused on MIMO systems without any quantization errors. 
Accordingly, antenna selection for low-resolution ADC systems that incorporates coarse quantization effect needs to be investigated.
In \cite{chen2019joint}, a cross entropy maximization approach in \cite{zhang2009receive} was extended for low-resolution ADC systems by jointly solving the user scheduling problem.
Transmit antenna selection was analyzed for single antenna selection by utilizing Weibul distribution in low-resolution ADC systems \cite{choi2019analysis}.
In \cite{choi2019analysis}, it was shown that although the TAS gain is limited when compared to the gain for perfect quantization, the TAS gain can still provide a large increase of ergodic rate.
Although the proposed receive antenna selection algorithm in \cite{chen2019joint} demonstrated its high performance, it can require high complexity when the number of candidate antennas are large due to its parameters such as the number of iterations and sampling.
In addition, the transmit antenna selection in \cite{choi2019analysis} considers single antenna selection and thus, it is difficult to be generalized to multiple antenna selection.

\subsection{Contributions}

In this paper, we extend our previous work \cite{choi2018antenna} to investigate antenna selection at a BS with a large number of antenna arrays in low-resolution ADC systems where both the BS and mobile stations (MSs) are equipped with low-resolution ADCs.
We investigate DL transmit antenna selection and UL receive antenna selection.
The contributions are summarized as follows:
\begin{itemize}[leftmargin=*]
	\item For narrowband channels, we show that the DL transmit antenna selection problem with zero-forcing (ZF) precoding in low-resolution ADC systems is equivalent to that in high-resolution ADC systems when antennas are selected to maximize the DL sum rate. 
	Observing the quantization effect in the SNR, we further analyze the DL sum rate with antenna selection by incorporating quantization effects.
	We show that selecting more transmit antennas provides larger maximum sum rate for low-resolution ADC systems as well as high-resolution ADC systems.
	Unlike the rate loss in high-resolution ADC systems, we prove that the rate loss decreases beyond a certain point of transmit power and converges to zero in low-resolution ADC systems. 
	\item For an UL receive antenna selection problem in the narrowband, we generalize an existing criterion for a greedy capacity-maximization antenna selection method to incorporate quantization effects.
	The derived objective function offers an opportunity to select an antenna with the best tradeoff between the additional channel gain and increase in quantization error. 
	We also derive a lower bound of the sum rate achieved by the proposed greedy algorithm by using a concept of submodularity.
	In addition, we modify the adaptive MCMC antenna selection \cite{liu2009low} for the low-resolution ADC systems to provide a numerical upper bound of the sum rate.
	\item We extend the antenna selection problem to the wideband OFDM systems. 
	We first derive the wideband OFDM systems under coarse quantization for both DL and UL communications. 
	Then, we show that the derived results in the DL narrowband  communications also hold for the DL OFDM communication when subcarriers share a common antenna subset.
	For the UL OFDM communications, we modify the proposed received antenna selection algorithms and derive the lower bound of the capacity with the greedy algorithm.
	\item Simulation results validate the theoretical results and demonstrate that the proposed algorithm outperforms conventional algorithms in achievable rate.
	The proposed receive antenna selection algorithm provides near optimal sum rate performance in the large antenna array regime.
\end{itemize}

{\it Notation}: $\bf{A}$ is a matrix and $\bf{a}$ is a column vector. 
$\mathbf{A}^{H}$ and $\mathbf{A}^T$  denote conjugate transpose and transpose. 
$[{\bf A}]_{i,:}$ and $ \mathbf{a}_i$ indicate the $i$th row and column vector of $\bf A$. 
We denote $a_{i,j}$ or $[\bA]_{i,j}$ as the $\{i,j\}$th element of $\bf A$ and $a_{i}$ as the $i$th element of $\bf a$. 
$\mathcal{CN}(\mu, \sigma^2)$ is the circularly complex Gaussian distribution with mean $\mu$ and variance $\sigma^2$. 
$\mathbb{E}[\cdot]$ and $\bbV[\cdot]$ represent an expectation and variance operators, respectively.
The correlation matrix is ${\bf R}_{\bf xy} = \mathbb{E}[{\bf x}{\bf y}^H]$.
The diagonal matrix $\rm diag\{\bf A\}$ has $\{a_{i,i}\}$ at its $i$th diagonal entry, and $\rm diag \{\bf a\}$ or ${\rm diag}\{{\bf a}^T\}$ has $\{a_i\}$ at its $i$th diagonal entry. 
${\rm BlkDiag}\{\bA_1,\dots,\bA_N\}$ is a block diagonal matrix with block diagonal entries $\bA_1,\cdots,\bA_N$.
${\rm BlkCirc}\{\bA_0, \bA_1,\cdots, \bA_N\}$ is a block circulant matrix with $[\bA_0, \bA_1,\cdots, \bA_N]$ at its first block row. 
${\bf I}_N$ is the $N\times N$ identity matrix and $\bf 0$ is a matrix that has all zeros in its entries with a proper dimension.
$\|\bf A\|$ represents $L_2$ norm. 
$|\cdot|$ indicates an absolute value, cardinality, and determinant for a scalar value $a$, a set $\cA$, and a matrix $\bA$, respectively.
A trace operator is ${\rm tr}\{\cdot\}$.

\section{System Model}
\label{sec:sys_model}

We consider single-cell multiuser systems in which a BS serves $N_{\rm MS}$ MSs.
As shown in Fig.~\ref{fig:system}, The BS is equipped with $N_{\rm BS}$ antennas and low-resolution ADCs.
Each MS is equipped with a single antenna and low-resolution ADCs.
We assume that the number of the BS antennas is much larger than the number of MSs, $N_{\rm BS} \gg N_{\rm MS}$.
The CSI is assumed to be known at the BS.

\begin{figure}[!t]\centering
	\includegraphics[scale = 0.45]{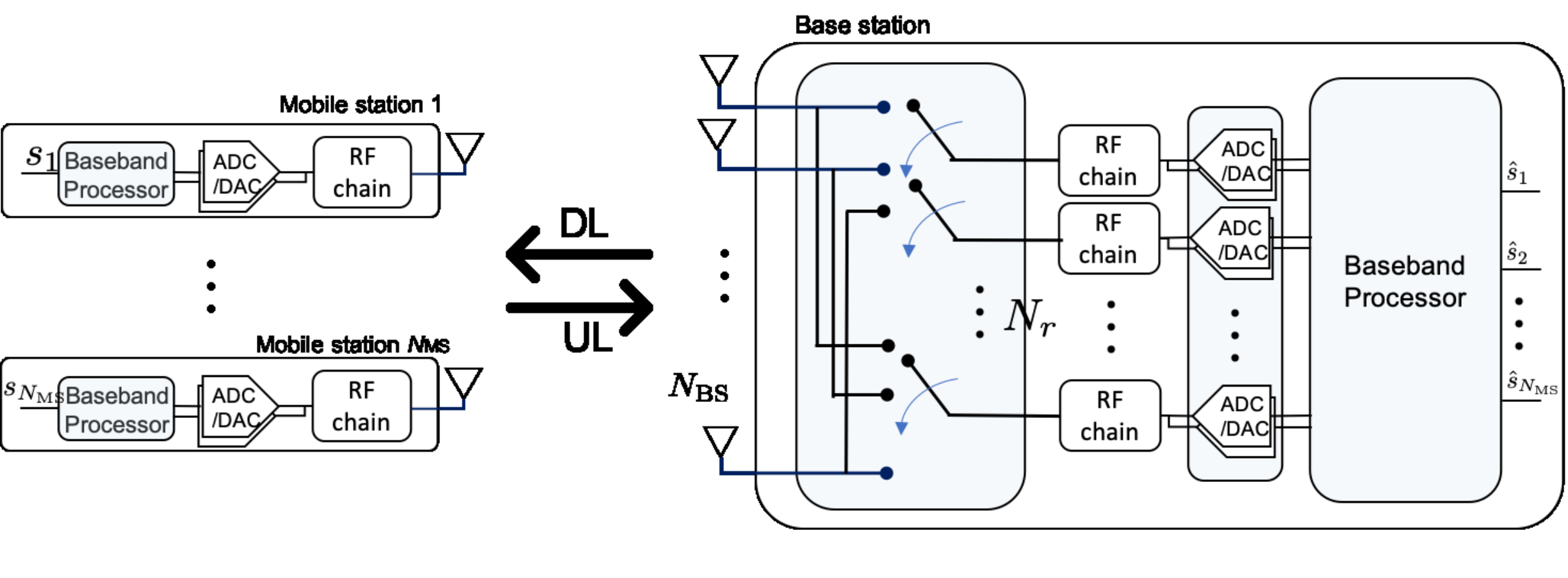}
	\caption{A multiuser communication system in which a base station (BS) serves $N_{\rm MS}$ mobile stations (MSs). The BS is equipped with $N_{\rm BS}$ antennas and low-resolution ADCs.	Each MS is equipped with a single antenna and low-resolution ADCs.} 
	\label{fig:system}
\end{figure}

\subsection{Downlink Narrowband System}
\label{subsec:downlink}

The BS selects $N_t$ transmit antennas and employs a ZF precoding to null multiuser interference signals by using the CSI.
The vector of the precoded transmit signals $\bx^{\rm dl}\in \bbC^{N_t}$ is given as 
\begin{align}
	\nonumber
	\bx^{\rm dl} = \bW_{\rm BB}(\cT)\bP^{1/2}\bs^{\rm dl}
\end{align} 
where $\bW_{\rm BB}(\cT)\in \bbC^{N_t\times N_{\rm MS}}$ is the precoder with the selected antennas in the subset of antenna indices $\cT$, $\bP = {\rm diag}\{p_1,\dots,p_{N_{\rm MS}}\}$ is the matrix of transmit power for $\bs^{\rm dl}$, and $\bs^{\rm dl} \in \bbC^{N_{\rm MS}}$ is the user symbol vector.
The transmit power is constrained by the total power constraint $P$ as
\begin{align}
	\label{eq:txpower_const_dl}
	{\rm tr}(\bbE[\bx^{\rm dl}\bx^{{\rm dl}\,H}]) = {\rm tr}(\bW_{\rm BB}(\cT)\bP\bW_{\rm BB}^H(\cT))  \leq P.
\end{align}
With ZF precoding, the precoder $\bW_{\rm BB}(\cT)$ becomes $\bW_{\rm BB}(\cT) = \bH_\mathcal{T}^{{\rm dl}\,H}(\bH_\mathcal{T}^{\rm dl}\bH_\mathcal{T}^{{\rm dl}\,H})^{-1}$.
Accordingly, the vector of received analog baseband signals at the MSs is given as 
\begin{align}
	\label{eq:r_dl_zf}
	\br^{\rm dl} =  \bH^{\rm dl}_\mathcal{T} \bx^{\rm dl} + \bn^{\rm dl}= \bP^{1/2}\bs^{\rm dl} + \bn^{\rm dl} 
\end{align}
where $ \bH^{\rm dl}_\mathcal{T} \in \bbC^{N_{\rm MS} \times N_t}$ is the  DL narrowband channel matrix, which consists of $N_t$ selected columns of the DL channel $\bH^{\rm dl}\in \bbC^{N_{\rm MS} \times N_{\rm BS}}$, and $\bn^{\rm dl} \sim \mathcal{CN}(\mathbf{0}, \mathbf{I}_{N_{\rm MS}})$ is the additive white circularly complex Gaussian noise (AWGN) vector.

Using the additive quantization noise model (AQNM) \cite{fletcher2007robust}, which provides a reasonable accuracy for low to medium SNR \cite{orhan2015low}, the quantized DL received signal vector is expressed as 
\begin{align}
	\nonumber
	 \by^{\rm dl}  &= \mathcal{Q}\bigl({\rm Re}\{{\br}^{\rm dl}\}\bigr)  + j\mathcal{Q}\bigl({\rm Im}\{{\br}^{\rm dl}\}\bigr)\\
	\label{eq:y_dl} 
	&=	\alpha_b \bP^{1/2} \bs^{\rm dl}  + \alpha_b  \bn^{\rm dl} +  \bq^{\rm dl} 
\end{align}
where $\mathcal{Q}(\cdot)$ is the element-wise quantizer function.
Here, $\alpha_b$ is defined as $\alpha_b = 1- \beta_b$ and considered to be the quantization gain $(\alpha_b <1)$, and $\beta_b$ is the normalized mean squared quantization error $\beta_b = \frac{\mathbb{E}[|{r}_i - {y}_i|^2]}{\mathbb{E}[|r_i|^2]}$.
Assuming a scalar minimum mean squared error (MMSE) quantizer and Gaussian signaling ${\bs}^{\rm dl} \sim \mathcal{CN}(\mathbf{0}, \mathbf{I}_{N_{\rm MS}})$, $\beta_b$ is approximated as $\beta_b \approx \frac{\pi\sqrt{3}}{2} 2^{-2b}$ for $b > 5$ \cite{gersho2012vector}, where $b$ is the number of quantization bits for each real and imaginary part.
The values of $\beta_b$ for $b \leq 5$ are shown in Table 1 in \cite{fan2015uplink}.
The vector ${\bq}^{\rm dl} \in \bbC^{N_{\rm MS}}$ represents the additive quantization noise that is uncorrelated with the quantization input ${\br}^{\rm dl}$ \cite{fletcher2007robust}.
We assume that the quantization noise follows the complex Gaussian distribution with a zero mean ${\bq}^{\rm dl} \sim \mathcal{CN}({\bf 0},{\bf R}_{\bq^{\rm dl} \bq^{\rm dl}})$ \cite{fan2015uplink}.
The covariance matrix of $\bq^{\rm dl}$ is derived as \cite{fan2015uplink}
\begin{align}
	\label{eq:Rqq_dl}
	\bR_{\bq^{\rm dl}\bq^{\rm dl}} = \alpha_b(1-\alpha_b){\rm diag}\big\{\bbE\big[\br^{\rm dl}\br^{{\rm dl}\,H}\big]\big\}= \alpha_b(1-\alpha_b)(\bP + \bI_{N_{\rm MS}}).
\end{align}


\subsection{Uplink Narrowband System}
\label{subsec:uplink}

The BS selects $N_r$ receive antennas and receives signals from $N_{\rm MS}$ MSs. 
The selected antennas are connected to RF chains followed by  low-resolution ADCs.
The UL narrowband channel matrix between the BS and MSs is denoted as $\bH^{\rm ul} \in \bbC^{N_{\rm BS}\times N_{\rm MS}}$.
The received baseband analog signals at the $N_r$ selected antennas ${\bf r}^{\rm ul} \in \mathbb{C}^{N_r}$ can be expressed as
\begin{align}
	\label{eq:r_ul}
	{\bf r}^{\rm ul} = \sqrt{\rho} {\bH}_{\mathcal{K}}^{\rm ul}{\bs}^{\rm ul} + {\bn}^{\rm ul}
\end{align} 
where $\rho$, $\bH_{\mathcal{K}}^{\rm ul}\in \bbC^{N_r \times N_{\rm MS}}$, ${\bs}^{\rm ul} \in \bbC^{N_{\rm MS}}$, and ${\bn}^{\rm ul} \in \bbC^{N_r}$ denotes the transmit power, the channel matrix for the selected antennas in the subset of antenna indices $\cK$, the user symbol vector, and the AWGN vector, respectively. 
We assume $\bs^{\rm ul}\sim \cC\cN({\bf 0},\bI_{N_{\rm MS}})$ and ${\bn}^{\rm ul}\sim \cC\cN({\bf 0},\bI_{N_{r}})$. 

After the antenna selection, each real and imaginary component of the complex output $r_i^{\rm ul}$, where $r_i^{\rm ul}$ denotes the $i$th element of $\br^{\rm ul}$ in \eqref{eq:r_ul}, is quantized at the pair of ADCs.
Adopting the AQNM \cite{fletcher2007robust}, the quantized UL received baseband signals becomes
\begin{align} 
	\nonumber
	{\by}^{\rm ul} & = \mathcal{Q}\bigl({\rm Re}\{{\br}^{\rm ul}\}\bigr)  + j\mathcal{Q}\bigl({\rm Im}\{{\br}^{\rm ul}\}\bigr)\\ 
	\label{eq:y_ul}
	&= \alpha_b \sqrt{\rho} {\bH}_\mathcal{K}^{\rm ul}{\bs}^{\rm ul}+ \alpha_b {\bn}^{\rm ul} +{\bq}^{\rm ul}
\end{align} 
where ${\bf q}^{\rm ul}$ represents the additive quantization noise that is uncorrelated with ${\bf r}^{\rm ul}$.
We assume  ${\bf q}^{\rm ul} \sim \mathcal{CN}({\bf 0},{\bf R}_{\bq^{\rm ul} \bq^{\rm ul}})$ \cite{fan2015uplink}.
The covariance matrix of $\bq^{\rm ul}$ is given by
\begin{align}
	\label{eq:Rqq_ul}
	\mathbf{R}_{\bq^{\rm ul}\bq^{\rm ul}}= \alpha_b(1-\alpha_b)\,{\rm diag}(\rho{\bH}_\mathcal{K}^{\rm ul}{\bH}_\mathcal{K}^{{\rm ul}\,H} + {\bI}_{N_r}).
\end{align}
In the following sections, we explore antenna selection for the considered DL and UL systems.

\section{Downlink Transmit Antenna Selection}
\label{sec:DL_txantenna}

In this section, we first show that a transmit antenna selection problem with ZF precoding for narrowband channels in low-resolution ADC systems is equivalent to that in high-resolution ADC systems.
The resulting achievable rate, however, involves the quantization error and thus, we analyze the sum rate in low-resolution ADC systems.

\subsection{Sum Rate Maximization Problem}
\label{subsec:problem_dl}

From the quantized signals $\by^{\rm dl}$ in \eqref{eq:y_dl} and quantization covariance matrix $\bR_{\bq^{\rm dl}\bq^{\rm dl}}$ in \eqref{eq:Rqq_dl}, the DL achievable rate for user $i$ with selected transmit antennas in $\cT$ becomes
\begin{align}
	\label{eq:rate_dl}
	\gamma_i^{\rm dl}(\cT)= \log_2\left(1 +\frac{\alpha_b^2 p_i}{\alpha_b^2 + \alpha_b(1-\alpha_b)(1+p_i)}\right).
\end{align}
We consider an equal power distribution. 
Assuming equal power distribution, $p_i = p_\mathcal{T}$, $\forall i$, and ZF precoding with maximum transmit power from \eqref{eq:txpower_const_dl}, we have 
\begin{align}
	\label{eq:txpower_dl}
	p_\mathcal{T} = \frac{P}{{\rm tr}\big(\bW_{\rm BB}^H(\cT)\bW_{\rm BB}(\cT)\big)} =  \frac{P}{{\rm tr}\big((\bH_\mathcal{T}\bH_\mathcal{T}^H)^{-1}\big)}.
\end{align}
Using \eqref{eq:rate_dl} and \eqref{eq:txpower_dl}, the DL achievable sum rate reduces to
\begin{align}
	\label{eq:sumrate_dl}
	\cR^{\rm dl}(\cT) = N_{\rm MS}\log_2\left(1+\frac{\alpha_b p_\mathcal{T}}{1+(1-\alpha_b)p_\mathcal{T}}\right). 
\end{align}  
We now formulate the transmit antennas selection problem by adopting the achievable sum rate in \eqref{eq:sumrate_dl} as an objective function.
Let $\cS = \{1,2,\dots, N_{\rm BS}\}$ be the index set of the BS antennas.
Then, the transmit antenna selection problem for maximum sum rate is formulated as
{\color{black}
\begin{align}
	\nonumber
	\cP 1:\quad\quad\cT^\star = \argmax_{\mathcal{T} \subseteq \mathcal{S}:N_{\rm MS} \leq|\mathcal{T}|\leq N_t } \cR^{\rm dl}(\cT).
\end{align}
where $N_t$ is the given maximal number of transmit antennas that can be selected.}
\begin{remark}
	The transmit antenna selection problem $\cP1$ with ZF precoding and equal power allocation for narrowband channels is equivalent to that in high-resolution ADC systems.
\end{remark}
Accordingly, we show that any state-of-the-art transmit antenna selection methods for multiuser communications with the ZF precoding \cite{khademi2014convex, lin2012performance} can be applicable  in low-resolution ADC systems.
The achievable rate $\cR^{\rm dl}(\cT)$, however, includes the quantization effect as a noise that is proportional to the transmit power, which differs from perfect quantization systems.  
In this regards, we provide theoretical analysis for the transmit antenna selection problem to characterize the sum rate and draw intuitions for the low-resolution ADC regime in the following subsection.

\subsection{Sum Rate Analysis of Transmit Antenna Selection}
\label{subsec:DL_analysis}

Here, we first derive a property of the sum rate in the considered low-resolution ADC system with respect to the number of selected antennas.
To this end, we introduce Lemma~\ref{lem:trace}.
\begin{lemma}
	\label{lem:trace}
	For any matrix $\bH \in \bbC^{m\times n}$ with ${\rm rank}(\bH) = m$, the following inequality holds:
	\begin{align}
		\nonumber
		{\rm tr}\left(\bQ\tilde{\bH}(\bI_\ell - \tilde{\bH}^H\bQ\tilde{\bH})^{-1}\tilde{\bH}^H\bQ\right) >0 
	\end{align}
	where $\bQ = (\tau\bI_m + \bH\bH^H)^{-1}$ with $\tau \geq 0$ and $\tilde{\bH}$ is a $m \times \ell $ sub-matrix of $\bH$ which consists of the columns of $\bH$ for $1\leq \ell \leq (n-m)$.
\end{lemma}
\begin{proof}
	See Lemma 2 in \cite{lin2012performance}.
\end{proof}
\begin{theorem}
	\label{thm:monotonic}
	The maximum sum rate of MSs with low-resolution ADCs in \eqref{eq:sumrate_dl} is monotonically increasing with the number of selected transmit antennas in ZF precoding DL systems \eqref{eq:r_dl_zf}:
	\begin{align}
		\nonumber
		\cR^{\rm dl}(\cT_{\rm opt1}) < \cR^{\rm dl}(\cT_{\rm opt2})
	\end{align}
	where $\cT_{\rm opt1}$ and $\cT_{\rm opt2}$ are the optimal antenna subsets with $|\cT_{\rm opt1}| < |\cT_{\rm opt2}|$.
\end{theorem}
\begin{proof}
	Let $\cT_1$ and $\cT_2$ be antenna subsets with $\cT_1 \subset	\cT_2 \subseteq \cS$, and $\bar\cT$  be $\bar\cT = \cT_2 - \cT_1$.
	The average sum rate difference between the sum rates with the two antenna subsets, $\cT_1$ and $\cT_2$, is
	\begin{align}
		\nonumber
		\frac{\cR^{\rm dl}_D(\bar\cT)}{N_{\rm MS}} &= \frac{\cR^{\rm dl}(\cT_2) - \cR^{\rm dl}(\cT_1)}{N_{\rm MS}} \\
		\label{eq:monotonic_proof1}
		& =  \log_2\left(1+\frac{\alpha_b p_{\mathcal{T}_2}}{1+(1-\alpha_b)p_{\mathcal{T}_2}}\right)- \log_2\left(1+\frac{\alpha_b p_{\mathcal{T}_1}}{1+(1-\alpha_b)p_{\rm \mathcal{T}_1}}\right).
	\end{align}
	Using $p_{\mathcal{T}_i}= {P}/{{\rm tr}((\bH_{\mathcal{T}_i}^{\rm dl}\bH_{\mathcal{T}_i}^{{\rm dl}\, H})^{-1})}$ for $i = 1, 2$, we rewrite \eqref{eq:monotonic_proof1} as 
	\begin{align}
		\nonumber
		\frac{\cR_D^{\rm dl}(\bar\cT)}{N_{\rm MS}} 
		= \log_2\left(\frac{({\rm tr}((\bH^{\rm dl}_{\mathcal{T}_2}\bH_{\mathcal{T}_2}^{{\rm dl}\, H})^{-1})+P)({\rm tr}((\bH_{\mathcal{T}_1}^{\rm dl}\bH_{\mathcal{T}_1}^{{\rm dl}\, H})^{-1})+(1-\alpha_b)P)}{({\rm tr}((\bH_{\mathcal{T}_2}^{\rm dl}\bH_{\mathcal{T}_2}^{{\rm dl}\, H})^{-1})+(1-\alpha_b)P)({\rm tr}((\bH_{\mathcal{T}_1}^{\rm dl}\bH_{\mathcal{T}_1}^{{\rm dl}\, H})^{-1})+P)}\right).
	\end{align}	
	Let $\bQ= (\bH_{\mathcal{T}_2}^{\rm dl}\bH_{\mathcal{T}_2}^{{\rm dl}\, H})^{-1}$ and ${\pmb \Psi}_{\bar{\mathcal{T}}} = \bQ\bH_{\bar{\mathcal{T}}}^{\rm dl}(\bI_{|\bar{\mathcal{T}}|} - \bH_{\bar{\mathcal{T}}}^{{\rm dl}\, H}\bQ\bH_{\bar{\mathcal{T}}}^{\rm dl})^{-1}\bH_{\bar{\mathcal{T}}}^{{\rm dl}\, H}\bQ$. 
	Then, leveraging the matrix inversion lemma, the rate difference  $\cR^{\rm dl}_D(\bar\cT)$, which we also call as the rate loss, becomes
	\begin{align}
		\nonumber
		&\cR^{\rm dl}_D(\bar\cT)
=\! N_{\rm MS} \log_2\!\left(\frac{({\rm tr}(\bQ)+P)({\rm tr}(\bQ)+{\rm tr}({\pmb \Psi}_{\bar{\mathcal{T}}})+(1-\alpha_b)P)} {({\rm tr}(\bQ)+(1-\alpha_b)P)({\rm tr}(\bQ)+{\rm tr}({\pmb \Psi}_{\bar{\mathcal{T}}})+P)}\right) \\
		\label{eq:monotonic_proof2}
		& = \!N_{\rm MS} \log_2\!\left(\!1 + \frac{\alpha_b {\rm tr}({\pmb \Psi}_{\bar {\mathcal{T}}})P}{{\rm tr}(\bQ)^2+\big({\rm tr}({\pmb \Psi}_{\bar {\mathcal{T}}}) + P\big){\rm tr}(\bQ) + (1-\alpha_b)\big(P^2 + P({\rm tr}({\pmb \Psi}_{\bar {\mathcal{T}}})+{\rm tr}(\bQ))\big)}\!\right)\\
		\label{eq:monotonic_proof3}
		& \stackrel{(a)}>0
	\end{align}
	where $(a)$ holds from the following reasons: we have ${\rm tr}(\bQ) > 0$, and from Lemma \ref{lem:trace} with $\tau = 0$, we have ${\rm tr}({\pmb \Psi}_{\bar{\mathcal{T}}}) > 0$ for any channel matrix $
	\bH_{\mathcal{T}_2}^{\rm dl}$ with ${\rm rank}(\bH_{\mathcal{T}_2}^{\rm dl}) = N_{\rm MS}
	$ and its $N_{\rm MS} \times |\bar{\mathcal{T}}|$ sub-matrix $\bH_{\bar{\mathcal{T}}}^{\rm dl}$ with $1\leq |\bar{\mathcal{T}}| \leq  (|\mathcal{T}_2| - N_{\rm MS})$.
	In addition, $\alpha_b$ is always less than one ($\alpha_b <1$) since it is the quantization gain defined as $\alpha_b = 1 -{\mathbb{E}[|{r}_i - {y}_{i}|^2]}/{\mathbb{E}[|{r}_i|^2]}$. 
	
	Now, let $\cT_2$ be the antenna subset that satisfies $\cT_{\rm opt1} \subset \cT_2$ and $|\cT_{\rm opt1}| < |\cT_2| = |\cT_{\rm opt2}|$. 
	Then, we obtain the following inequalities: 
	\begin{align}
		\nonumber
		\cR^{\rm dl}(\cT_{\rm opt1}) < \cR^{\rm dl}(\cT_2) \leq \cR^{\rm dl}(\cT_{\rm opt2})
	\end{align}
	where $\cR^{\rm dl}(\cT_{\rm opt1}) < \cR^{\rm dl}(\cT_2)$ follows from leveraging $\cR^{\rm dl}_D(\bar\cT) > 0 $ in \eqref{eq:monotonic_proof2} and $\cR^{\rm dl}(\cT_2) \leq \cR^{\rm dl}(\cT_{\rm opt2})$ comes from the optimality definition of $\cT_{\rm opt2}$. 
	This completes the proof.
\end{proof}
Although adding more transmit antennas is not guaranteed to increase the sum rate \cite{vaze2012submodularity} in general because of a transmit power constraint, Theorem \ref{thm:monotonic} shows that the maximum sum rate increases with the number of selected transmit antennas $N_t$ even with the coarse quantization at the user mobile.
This result was also shown to be true for high-resolution ADC systems  \cite{lin2012performance}.
Now we will show that the sum rate loss $\cR^{\rm dl}_{D}(\bar \cT)$ has a different property compared to the high-resolution ADC systems where the loss monotonically increases with $P$ and converges to an upper bound \cite{lin2012performance}. 
Having $\cT_2 = \cS$, $\cR^{\rm dl}_D(\bar{\mathcal{T}})$ can be considered as the sum rate loss due to antennas selection and minimized to zero by increasing the transmit power constraint $P$.
\begin{corollary}
	\label{cor:converge}
	Let $\cT_1 \subset \cT_2 \subseteq \cS$, then the achievable sum rate loss $\cR^{\rm dl}_D(\bar{\cT}) = \cR^{\rm dl}(\cT_2) - \cR^{\rm dl}(\cT_1)$ goes to zero under coarse quantization as the transmit power constraint $P$ increases
	\begin{align}
		\nonumber
		\cR^{\rm dl}_D(\bar{\cT}) \to 0 \quad \text{as } P \to \infty.
	\end{align}
	In addition, the achievable rate converges to $\cR^{\rm dl}(\cT) \to N_{\rm MS} \log_2\left(1+\frac{\alpha_b}{1-\alpha_b}\right)$ as $P \to \infty$.
\end{corollary}
\begin{proof}
 	If $P \to \infty$, the achievable sum rate loss in \eqref{eq:monotonic_proof2} goes to zero and the sum rate in \eqref{eq:sumrate_dl} converges to $N_{\rm MS} \log_2\left(1+\frac{\alpha_b}{1-\alpha_b}\right)$.
\end{proof}

Unlike the high-resolution ADC system, this result suggests that antenna selection can have the marginal rate loss from the system using the entire antennas by increasing $P$.
\begin{corollary}
	\label{cor:RLmax}
	Let $\cT_1 \subset \cT_2 \subseteq \cS$. 
	Then, the transmit power constraint that leads to the maximum sum rate loss from not using antennas in $\bar{\cT} = \cT_2 - \cT_1$ is 
	\begin{align}
		\label{eq:Pmax}
		P_D^{\rm max} = \sqrt{\frac{ {\rm tr}(\bQ){\rm tr}(\bK)}{1-\alpha_b}}
	\end{align}
	where $\bQ = (\bH^{\rm dl}_{\mathcal{T}_2}\bH^{{\rm dl}\,H}_{\mathcal{T}_2})^{-1}$ and $\bK = (\bH^{\rm dl}_{\mathcal{T}_1}\bH^{{\rm dl}\,H}_{\mathcal{T}_1})^{-1}$, and the maximum sum rate loss is
	\begin{align}
		\label{eq:RLmax}
		&\cR_D^{\rm dl, max}(\bar{\mathcal{T}}) =  N_{\rm MS}\log_2\left(1+\frac{\alpha_b \big({\rm tr}(\bK)-{\rm tr}(\bQ)\big)}{{\rm tr}(\bQ)+(1-\alpha_b){\rm tr}(\bK)+ 2\sqrt{(1-\alpha_b){\rm tr}(\bQ){\rm tr}(\bK)}}\right).	
	\end{align}
\end{corollary}
\begin{proof}
	Let $\bQ= (\bH^{\rm dl}_{\mathcal{T}_2}\bH_{\mathcal{T}_2}^{{\rm dl}\,H})^{-1}$ and ${\pmb \Psi}_{\bar{\mathcal{T}}} = \bQ\bH^{\rm dl}_{\bar{\mathcal{T}}}(\bI_{|\bar{\mathcal{T}}|} - \bH_{\bar{\mathcal{T}}}^{{\rm dl}\,H}\bQ\bH^{\rm dl}_{\bar{\mathcal{T}}})^{-1}\bH_{\bar{\mathcal{T}}}^{{\rm dl}\,H}\bQ$. 
	The derivative of \eqref{eq:monotonic_proof2} with respect to the transmit power constraint is derived as 
	\begin{align}
		\label{eq:dRL}
		\frac{d \cR^{\rm dl}_D(\bar{\cT})}{d P} = \frac{\alpha_b N_{\rm MS}{\rm tr}(\pmb \Psi_{\bar{\mathcal{T}}})\big({\rm tr}(\bQ)^2 + {\rm tr}(\bQ){\rm tr}({\pmb \Psi_{\bar{\mathcal{T}}}})+(\alpha_b -1)P^2\big)}{\Gamma_{\bar{\mathcal{T}}}}
	\end{align}
	where $ \Gamma_{\bar{\mathcal{T}}} = \ln2({\rm tr}(\bQ)+P)({\rm tr}(\bQ)+{\rm tr}({\pmb \Psi_{\bar{\mathcal{T}}}})+P)({\rm tr}(\bQ)+(1-\alpha_b)P)({\rm tr}(\bQ)+{\rm tr}({\pmb \Psi_{\bar{\mathcal{T}}}})+(1-\alpha_b)P)$.
	Since $0<\alpha_b <1$ and ${\rm tr}(\pmb \Psi_{\bar{\mathcal{T}}}) >0$,
	 by setting \eqref{eq:dRL} to be zero, we derive $P^{\rm max}_D$ as
	\begin{align}
		\label{eq:Pmax_proof1}
		P_D^{\rm max} = \sqrt{\frac{ {\rm tr}(\bQ)^2 + {\rm tr}(\bQ){\rm tr}(\pmb \Psi_{\bar{\mathcal{T}}})}{1-\alpha_b}}.	
	\end{align}
	Using ${\rm tr}\big((\bH^{\rm dl}_{\mathcal{T}_1}\bH_{\mathcal{T}_1}^{{\rm dl}\,H})^{-1}\big) = {\rm tr}(\bQ) + {\rm tr}(\pmb \Psi_{\bar{\mathcal{T}}})$, the maximizer $P_D^{\rm max}$ \eqref{eq:Pmax_proof1} is rewritten as \eqref{eq:Pmax}.
	With respect to the transmit power constraint $P$, the maximum sum rate loss for $\cT_1$ and $\cT_2$ can be determined by putting $P = P_D^{\rm max}$ into \eqref{eq:monotonic_proof3}, which leads to \eqref{eq:RLmax}. 
	This completes the proof.
\end{proof}
According to Corollary~\ref{cor:RLmax}, the transmit antenna selection in low-resolution ADC systems always achieves the sum rate with the rate loss less than $\cR_D^{\rm dl, max}(\bar{\mathcal{T}})$ in \eqref{eq:RLmax} for a selected antenna subset.
Note that if there is no quantization error, i.e., $\alpha_b = 1$, $P_D^{\rm max}$ goes to infinity. 
Then, the sum rate loss cannot decrease with $P$ in the perfect quantization system, which corresponds to the upper bound of the sum rate loss in \cite{lin2012performance}.
Since $ \Gamma_{\bar{\mathcal{T}}}$ and ${\rm tr}(\pmb \Psi_{\bar{\mathcal{T}}})$ are positive,  ${\partial \cR^{\rm dl}_D(\bar{\cT})}/{\partial P}$ in \eqref{eq:dRL} becomes positive when $P < P_D^{\rm max}$ and negative when $P > P_D^{\rm max}$, i.e., for $P < P_D^{\rm max}$, the sum rate loss increases as $P$ increases, and for $P > P_D^{\rm max}$, the loss decreases to zero as $P$ increases. 
Therefore, \eqref{eq:Pmax} can be considered as the reference power constraint that is required to reduce the sum rate loss while achieving a reasonable sum rate. 
\begin{corollary}
	\label{cor:RlmaxvsRLmax}
	The maximum rate loss in low-resolution ADC systems is less than that in high-resolution ADC systems, i.e., $\cR_D^{\rm dl, max}(\bar{\cT};b) \leq \cR_D^{\rm dl, max}(\bar{\cT};\infty).$
\end{corollary}
\begin{proof}
	Since ${\rm tr}(\pmb \Psi_{\bar{\mathcal{T}}}) = {\rm tr}(\bK) - {\rm tr}(\bQ) > 0$ from Lemma \ref{lem:trace}, where $\bQ = (\bH^{\rm dl}_{\mathcal{T}_2}\bH^{{\rm dl}\,H}_{\mathcal{T}_2})^{-1}$, $\bK = (\bH^{\rm dl}_{\mathcal{T}_1}\bH^{{\rm dl}\,H}_{\mathcal{T}_1})^{-1}$, and ${\pmb \Psi}_{\bar{\mathcal{T}}} = \bQ\bH^{\rm dl}_{\bar{\mathcal{T}}}(\bI_{|\bar{\mathcal{T}}|} - \bH_{\bar{\mathcal{T}}}^{{\rm dl}\,H}\bQ\bH^{\rm dl}_{\bar{\mathcal{T}}})^{-1}\bH_{\bar{\mathcal{T}}}^{{\rm dl}\,H}\bQ$, the maximum rate loss in \eqref{eq:RLmax} is a monotonically increasing function with respect to $\alpha_b$ with $0<\alpha_b <1$. 
	When $\alpha_b  \to 1$, the considered system becomes equivalent to the high-resolution ADC system.
\end{proof}
Based on Corollary~\ref{cor:RlmaxvsRLmax}, the transmit antenna selection can be more effective in low-resolution ADC systems as the rate loss is smaller than that in high-resolution ADC systems.

\section{Uplink Receive Antenna Selection}
\label{sec:UL_rxantenna}

In this section, we examine the key difference of the receive antenna selection problem at the BS with low-resolution ADCs from the conventional problem and propose a quantization-aware receive antenna selection method.

\subsection{Capacity Maximization Problem}
\label{subsec:problem}

For the considered UL narrowband system in \eqref{eq:y_ul}, the capacity can be expressed as
\begin{align}
	\label{eq:capacity}
	\cR^{\rm ul}(\cK) = \log_2 \Big|{\bf I}_{N_r} + \rho\alpha_b^2\big( \alpha_b^2 {\bf I}_{N_r} + {\bf R}_{\bq^{\rm ul}\bq^{\rm ul}}\big)^{-1}{\bf H}^{\rm ul}_\mathcal{K}{\bf H}^{{\rm ul}\,H}_\mathcal{K} \Big|
\end{align}
where ${\bf R}_{\bq^{\rm ul}\bq^{\rm ul}}$ is given in \eqref{eq:Rqq_ul}.
We note from \eqref{eq:capacity} that in the low-resolution ADC system, the capacity  involves the quantization noise covariance matrix ${\bf R}_{\bq^{\rm ul}\bq^{\rm ul}}$ as a penalty term for each antenna.
We use ${\bf f}^H_i$ to indicate the $i$th row of ${\bH}^{\rm ul}$ and ${\mathcal{K}(i)}$ to denote the $i$th selected antenna.
\begin{remark}
	\label{rm:intuition}
	Since each diagonal entry of ${\bf R}_{ \bq^{\rm ul}\bq^{\rm ul}}$ contains an aggregated channel gains at each selected antenna $\|{\bf f}_{\mathcal{K}(i)}\|^2$, the tradeoff between the channel gain from adding antennas and its influence on quantization error needs to be considered in antenna selection.
\end{remark}

Using the capacity in \eqref{eq:capacity}, we formulate the UL receive antenna selection problem as follows:
\begin{align}
	\label{eq:problem_ul}
	\cP 2:\quad\quad\mathcal{K}^\star = \argmax_{\mathcal{K} \subseteq \mathcal{S}:|\mathcal{K}| = N_r \geq N_{\rm MS}}  \cR^{\rm ul}\big (\cK\big),
\end{align}
where $\cS = \{1,\dots,N_{\rm BS}\}$. 
Notice that the large number of BS antennas $N_{\rm BS}$ makes it almost infeasible to perform an exhaustive search. 
Accordingly, to avoid searching over all possible antenna subsets $\mathcal{K}$, we propose two algorithms: a quantization-aware antenna selection algorithm based on the greedy approach and a Markov chain Monte Carlo (MCMC)-based algorithm.

\subsection{Greedy Approach}
\label{subsec:greedy}

Now, let ${\bf D}_{\mathcal{K}} = {\rm diag}\{1+\rho(1-\alpha_b)\|{\bf f}_{\mathcal{K}(i)}\|^2\}$ be the diagonal matrix with $(1+\rho(1-\alpha_b)\|{\bf f}_{\mathcal{K}(i)}\|^2)$ for $i = 1, \dots, N_r$ at its diagonal entries.
Then, the capacity in \eqref{eq:capacity} can be rewritten as 
\begin{align}
    &\cR^{\rm ul}(\cK)
	\label{eq:capacity2}
	 =  \log_2 \Big|{\bf I}_{N_r} +  \rho\alpha_b{\bf D}_{\mathcal{K}}^{-1}{\bf H}_{\mathcal{K}}^{\rm ul}{\bf H}^{{\rm ul}\,H}_{\mathcal{K}} \Big|.
\end{align}
Let $\cK_t$ be the set of selected antennas during the first $t$ greedy selections and $\bH_{\mathcal{K}_t\cup \{j\}}$ be the channel matrix of $t$ selected antennas during the first $t$ greedy selections and a candidate antenna $j\in \mathcal{S}\setminus \mathcal{K}_t$ at the next selection stage. 
Then, we formulate the greedy selection problem as
\begin{align}
	\label{eq:greedy-max}
	J = \argmax_{j \in \mathcal{S}\setminus \mathcal{K}_t} \cR^{\rm ul}(\cK_t\cup \{j\}).
\end{align}
To reduce the complexity of solving the problem in \eqref{eq:greedy-max}, we decompose the capacity formula \eqref{eq:capacity2}.
At the $(t+1)$th selection stage with a candidate antenna $j$, we have
\begin{align}
	\nonumber
    \cR^{\rm ul}(\cK_{t}\cup \{j\}) & = \log_2 \Big|{\bf I}_{N_r} + \rho\alpha_b {\bf D}^{-1}_{{\mathcal{K}_t\cup \{j\}}}{\bf H}^{\rm ul}_{{\mathcal{K}_t\cup \{j\}}}{\bf H}^{{\rm ul}\,H}_{{\mathcal{K}_t\cup \{j\}}} \Big|\ \\ 
	\label{eq:capacity3}
	& = \log_2 \biggl|{\bf I}_{N_{\rm MS}} + \rho\alpha_b\Bigl({\bf H}^{{\rm ul}\,H}_{\mathcal{K}_t} {\bf D}^{-1}_{\mathcal{K}_t}{\bf H}^{\rm ul}_{\mathcal{K}_t}\! +\! \frac{1}{d_{j}}{\bf f}_{j}{\bf f}^H_{j}\Bigr)\biggr|.
\end{align}
Recall that ${\bf f}^H_{j}$ denotes the $j$th row of ${\bf H}^{\rm ul}$ and $d_{j}$ is the corresponding diagonal entry of ${\bf D}_{{\mathcal{K}_t\cup \{j\}}}$.

Using the matrix determinant lemma $|{\bf A} + {\bf u}{\bf v}^H| = |{\bf A}|(1+{\bf v}^H{\bf A}^{-1}{\bf u})$, we rewrite \eqref{eq:capacity3} as
\begin{align}
	\cR^{\rm ul}(\cK_t\cup \{j\}) 
	\label{eq:rate_reduced_ul}
	= \cR^{\rm ul}(\cK_t) + \log_2\biggl(1\!+\!\frac{\rho \alpha_b}{d_j}c_{t}(j) \biggr)
\end{align}
where 
\begin{align}
	\label{eq:capacity gain}
	c_{t}(j) ={\bf f}^H_{j} \Bigl({\bf I}_{N_{\rm MS}}\! +\! \rho\alpha_b{\bf H}^{{\rm ul}\,H}_{\mathcal{K}_t} {\bf D}^{-1}_{\mathcal{K}_t}{\bf H}^{\rm ul}_{\mathcal{K}_t}\Bigr)^{-1}{\bf f}_{j}.
\end{align}
To maximize $\cR^{\rm ul}({\bf H}_{{\mathcal{K}_t\cup \{j\}}})$ given the $t$ selected antennas, the next antenna $j$ which maximizes ${c_{t}(j)}/{d_j}$ needs to be selected at the $(t+1)$th selection stage as 
\begin{align}
	\label{eq:problem2}
    J =\argmax_{j \in \mathcal{S}\setminus \mathcal{K}_t}  \frac{c_{t}(j)}{d_j}. 
\end{align}
Unlike the criterion with no quantization error in \cite{gharavi2004fast}, the derived criterion $c_{t}(j)/d_j$ incorporates $(i)$ the effect of the existing quantization error from the previously selected $t$ antennas to the next antenna $j$ in $c_{t}(j)$, and $(ii)$ the additional quantization error from the antenna $j$ as a penalty for selecting the antenna $j$ in the form of $1/d_j$.
In this regard, solving the problem \eqref{eq:problem2} gives the antenna $J$ which offers the best tradeoff between the channel gain from selecting an antenna and its influence on the increase of the quantization error.
We note that \eqref{eq:problem2} is the generalized antenna selection criterion of the one in \cite{gharavi2004fast}; as the number of quantization bits $b$ increases, the quantization gain $\alpha_b$ increases as $\alpha_b \to 1$, which leads to $d_j \to 1$ and ${\bf D}_{\mathcal{K}_t} \to {\bf I}_t$.

\begin{algorithm}[t!]
\label{algo:QFAS}
 {\bf Initialization}: $\mathcal{S} = \{1,\dots,N_{\rm BS}\}$, $\cK = \emptyset$ and ${\bf Q} = {\bf I}_{N_{\rm MS}}$.\\
 Compute initial antenna gain and compute penalty: \\
		$c(j) = \|{\bf f}_j\|^2$ and $d_j = 1 + \rho(1-\alpha_b)\|{\bf f}_j\|^2$ for $j \in \mathcal{S}$.\\
 \For{$t = 1:N_r$}{
    Select antenna $J$ using \eqref{eq:problem2}: $J=  \argmax_{j\in \mathcal{S}}  c(j)/d_j$.\\
    Update sets: $\mathcal{S} = \mathcal{S}\setminus\{J\}$ and $\cK = \cK\cup\{J\}$\\
    Compute: ${\bf a} = \bigl({c({J}) + \frac{d_J}{\rho \alpha_b}}\bigr)^{-\frac{1}{2}} {\bf Q} {\bf f}_J$ and ${\bf Q} = {\bf Q} - {\bf a}{\bf a}^H$.\\
    Update $c(j) = c(j) - |{\bf f}^H_j{\bf a}|^2$ for $j \in \mathcal{S}$.
      }
\Return{\ }{$\cK$;}
\caption{Quantization-aware Fast Antenna Selection (QFAS)}
\end{algorithm}

We now propose a quantization-aware fast antenna selection (QFAS) algorithm by using the derived criterion in \eqref{eq:problem2} and modifying the selection algorithm in \cite{gharavi2004fast} without increasing the overall complexity.
Unlike the perfect quantization case, the quantization error term $d_j$ needs to be computed prior to selection.
At each selection stage, the proposed algorithm adopts \eqref{eq:problem2}.
To compute $c_{t}(j)$ in \eqref{eq:capacity gain}, we define ${\bf Q}_t =  \Bigl({\bf I}_{N_{\rm MS}}\! +\! \rho\alpha_b{\bf H}^H_{\mathcal{K}_t} {\bf D}^{-1}_{\mathcal{K}_t}{\bf H}_{\mathcal{K}_t}\Bigr)^{-1}.$
Then, $c_{t}(j)$ is updated as
\begin{align}
	\nonumber
	c_{t+1}(j) = {\bf f}^H_{j}{\bf Q}_{t+1}{\bf f}_{j}
	\stackrel{(a)}= c_{t}(j) - |{\bf f}^H_j{\bf a}|^2.
\end{align}
where $(a)$ follows from that ${\bf Q}_t$ can be efficiently updated by using the matrix inversion lemma as ${\bf Q}_{t+1} = {\bf Q}_t - {\bf a}{\bf a}^H$ with ${\bf a} = \bigl({c_t({J}) + \frac{d_J}{\rho \alpha_b}}\bigr)^{-{1}/{2}} {\bf Q}_t {\bf f}_J$. 
The proposed QFAS algorithm is described in Algorithm \ref{algo:QFAS}.
Note that the complexity for step 5 and 6 are $\cO(N_rN_{\rm MS}^2)$ and $\cO(N_rN_{\rm MS} N_{\rm BS})$, respectively.
The overall complexity becomes $\cO(N_r N_{\rm MS} N_{\rm BS})$ because of $(N_{\rm BS}\gg N_{\rm MS})$.
Thus, the proposed algorithm does not increase the overall complexity from the conventional algorithm \cite{gharavi2004fast}, which provides the opportunity to be practically implemented.

Now, we analyze the performance of the proposed QFAS method by using submodularity.
\begin{definition}[Submodularity]
	If $\cV$ is a finite set, a submodular function is a set function $f:2^{\mathcal{V}}\to \bbR$ which meets the following condition: for every $\cA, \cB  \subseteq  \cV$ with $\cA \subseteq \cB$ and every element $v \in \cV \setminus \cB$, $f$ satisfies that $f(\cA \cup \{v\}) - f(\cA) \geq f(\cB \cup \{v\}) - f(\cB)$.
\end{definition}
\begin{definition}[Monotone]
	A set function $f:2^{\mathcal{V}}\to \bbR$ is monotone if for every $\cA \subseteq \cB \subseteq \cV$, we have that $f(\cA) \leq f(\cB)$.
	{\color{black} $f$ is said to be normalized if $f(\phi)=0$, where $\phi$ denotes the empty set.}
\end{definition}
From the definition of a submodular set function, it exhibits a diminishing return property.
The following theorem provides a performance lower bound of greedy methods for optimizing submodular objective functions. 
\begin{theorem}[\cite{nemhauser1978analysis}]\label{thm:submodular}
	{\color{black} For a normalized nonnegative and monotone submodular function} $f:2^{\mathcal{V}}\to \bbR_{+}$, let $\cA_{\rm G} \subseteq \cV$ be a set with $|\cA_{\rm G}| = k$ obtained by selecting elements one at a time and choosing an element that provides the largest marginal increase in the function value at each time.
	Let $\cA^\star$ be the optimal set that maximizes the value of $f$ with $|\cA^\star| = k$.
	Then, $f(\cA_{\rm G}) \geq (1-\frac{1}{e})f(\cA^\star)$. 
\end{theorem}
Based on Theorem \ref{thm:submodular}, it was shown in \cite{vaze2012submodularity} that the achievable rate of a point-to-point MIMO system is a submodular function, and hence, the greedy antenna selection algorithm for high-resolution ADC systems provides at least $\left(1-\frac{1}{e}\right)\cR^{\rm opt}$, where $\cR^{\rm opt}$ the achievable rate with the optimal antenna subset for high-resolution ADC systems.
We extend this result to the capacity with the quantization error in \eqref{eq:capacity}.
\begin{corollary}\label{cor:submodular_flat}
	The capacity achieved by the proposed QFAS method is lower bounded by
	\begin{align}
		\label{eq:greedy_LB}
		\cR^{\rm ul}({\cK}_{\rm qfas}) \geq \left(1-\frac{1}{e}\right)\cR^{\rm ul}({\cK}^{\star}).
	\end{align}
\end{corollary}
\begin{proof}
	We first need to show that the achievable rate with the quantization error $\cR^{\rm ul}(\cK)$ in \eqref{eq:capacity} is submodular.
	{\color{black} Let ${\pmb \Gamma}_{\mathcal{K}} = {\bf I}_{N_r} + \rho\alpha_b^2\big( \alpha_b^2 {\bf I}_{N_r} + {\bf R}_{{\bf q}^{\rm ul}{\bf q}^{\rm ul}}\big)^{-1/2}{\bf H}^{\rm ul}_\mathcal{K}{\bf H}^{{\rm ul}\,H}_\mathcal{K}\big(\alpha_b^2 {\bf I}_{N_r} + {\bf R}_{{\bf q}^{\rm ul}{\bf q}^{\rm ul}}\big)^{-1/2}$.} 
	Let $\bx_{\mathcal{K}} \sim \cC\cN({\bf 0}, {\pmb \Gamma}_{\mathcal{K}})$. 
	Since ${\pmb \Gamma}_{\mathcal{K}}$ is nonsingular, the entropy of ${\bf x}_{\mathcal{K}}$ is given as  
	\begin{align}
		\nonumber
		h({\bx}_\mathcal{K}) = \ln|\pi e {\pmb \Gamma}_{\mathcal{K}}|= N_r\ln(\pi e) + \frac{1}{\log_2 e}\cR^{\rm ul}(\cK).
	\end{align}
	{\color{black}Exploiting the form of $\bR_{\bq^{\rm ul},\bq^{\rm ul}}$ in \eqref{eq:Rqq_ul}}, for any sets $\cA \subseteq \cB \subseteq \cS$ and element such that $\{s\}\notin \cB$ and $\{s\}\in \cS$, we have $h({\bx}_{\{s\}}|{\bx}_{\mathcal{A}}) \geq h(\bx_{\{s\}}|{\bx}_{\mathcal{B}})$, i.e., $h({\bx}_{\mathcal{A}\cup\{s\}})-h({\bx}_{\mathcal{A}}) \geq h({\bx}_{\mathcal{B}\cup\{s\}})-h({\bx}_{\mathcal{B}})$.
	The entropy is submodular and $\cR^{\rm ul}(\cK)$ in \eqref{eq:capacity} is also submodular.
	{\color{black} In addition, $\cR^{\rm ul}(\cK)$ is normalized} and monotone.
	Since $\cR^{\rm ul}(\cK)$ \eqref{eq:capacity} is submodular, monotone, and nonnegative, the capacity with the greedy maximization in \eqref{eq:greedy-max} is lower bounded by \eqref{eq:greedy_LB} from Theorem~\ref{thm:submodular}. 
	Thus, the capacity with the proposed QFAS is also lower bounded by \eqref{eq:greedy_LB}.
\end{proof}

\subsection{Markov Chain Monte Carlo Approach}
\label{subsec:mcmc}

To find a numerical upper bound of the capacity for the antenna selection without exhaustive search, we provide an algorithm that finds an approximated optimal solution for the problem $\cP2$ in \eqref{eq:problem_ul}. 
We modify the adaptive MCMC-based selection method \cite{liu2009low} by adopting \eqref{eq:capacity} for formulating an original probability density function (PDF).
To develop the MCMC-based algorithm for low-resolution ADC systems, we define a binary vector $\pmb \omega \in \{0,1\}^{N_{\rm BS}}$ with $\|\pmb \omega\|_0 = N_r$ where $1$ indicates that the corresponding receive antenna is selected and vice versa.
Here, $\pmb \omega$ can be considered as a codeword of the codebook $\cV$ that contains all possible combinations of antenna subsets of size $N_r$, i.e.,  $|\cV| = {N_{\rm BS}\choose N_r}$.
Now, let the original PDF be
\begin{align}
	\label{eq:originalPDF}
	\pi(\pmb \omega) \triangleq  \exp\left(\frac{1}{\tau}{\mathcal{R}^{\rm ul}(\pmb \omega)}\right)/\Gamma
\end{align}
where $\tau$ is a rate constant and $\Gamma$ is a normalizing factor.
We  reformulate $\cP2$ in \eqref{eq:problem_ul} as
\begin{align}
	\label{eq:problem_ul_mcmc}
	\pmb \omega^\star = \argmax_{\pmb \omega \in \mathcal{V}} \pi(\pmb \omega).
\end{align} 

To solve \eqref{eq:problem_ul_mcmc}, the proposed algorithm uses a Metropolized independence sampler (MIS) \cite{liu2008monte} for the MCMC sampling, which is performed as follows: for a given current sample $\pmb \omega(i)$, a new sample $\pmb \omega^{\rm new}$ is selected according to a proposal distribution $q(\pmb \omega)$.
Based on a accepting probability $p_{\rm accept}(\pi, q) = {\rm min}\{1, \frac{\pi(\pmb \omega^{\rm new})}{\pi(\pmb \omega(i))}\frac{q(\pmb \omega(i))}{q(\pmb \omega^{\rm new})}\}$, we obtain a next sample as $\pmb \omega(i+1) = \pmb \omega^{\rm new}$ if accepted, or we have $\pmb \omega(i+1) = \pmb \omega(i)$, otherwise.
{\color{black} After $N_{\rm MCMC}$ iterations, we have a set of $(1+N_{\rm MCMC})$ samples including an initial sample $\pmb \omega(0)$, i.e., $\{\pmb \omega(0),\pmb \omega(1),\dots ,\pmb \omega(N_{\rm MCMC})\}$.}

For the proposal distribution, we use the product of Bernoulli distributions which is given as
\begin{align}
	\label{eq:proposalPDF}
	q(\pmb \omega; {\bf p}) = \frac{1}{\Gamma'} {\prod_{j=1}^{N_{\rm BS}}p_j^{ [\pmb\omega_{v}]_j}(1-p_j)^{1- [\pmb\omega_{v}]_j}}
\end{align}
where $p_j$ represents the probability of receive antenna $j$ to be selected and $[\pmb \omega]_j$ denotes the $j$th element of $\pmb \omega$.
Since $\Gamma'$ is unnecessary for computing the accepting probability $p_{\rm accept}$, we use $q(\pmb \omega; {\bf p})$ without  $\Gamma'$.
Similarly, we also use $\pi(\pmb \omega)$ without the normalizing factor $\Gamma$ for $p_{\rm accept}$.  

The selection probabilities $\bp$ will be adaptively updated at each iteration in the algorithm to increase the similarity between $\pi(\pmb \omega)$ and $q(\pmb \omega; {\bf p})$.
We update the probability entries $p_j$ to update the proposal distribution $q(\pmb \omega;\bp)$ by minimizing the Kullback-Leibler divergence between $\pi(\pmb \omega)$ and $q(\pmb \omega; {\bf p})$ \cite{liu2009low}.
Then, the update at $(t+1)$th iteration becomes
\begin{align}
	\label{eq:p_update}
	p_j^{(t+1)} = p_j^{(t)} + r^{(t+1)}\left(\frac{1}{N_{\rm MCMC}}\sum_{i=1}^{N_{\rm MCMC}}\left[\pmb \omega{(i)}\right]_j - p_j^{(t)}\right)
\end{align}
where $r^{(t)}$ is a sequence of decreasing step sizes that satisfies $\sum_{t=0}^\infty r^{(t)} = \infty$ and $\sum_{t=0}^\infty (r^{(t)})^2 < \infty$ \cite{harold1997stochastic}.
Finally, Algorithm~\ref{algo:QMCMC} describes the quantization-aware MCMC-based antenna selection (QMCMC-AS) algorithm. 
Algorithm \ref{algo:QMCMC} stops once it reaches a stopping criterion, which we set as the number of maximum iterations $\tau_{\rm stop}$. 
The computational complexity of the QMCMC-AS method is $\cO(N_rN_{\rm MS}^2N_{\rm MCMC}\tau_{\rm stop})$ \cite{liu2009low}.
We note that unlike the QFAS method, the complexity of the QMCMC-AS method involves additional parameters such as the sample size $N_{\rm MCMC}$ and the number of iterations $\tau_{\rm stop}$. 
When $N_{\rm BS}\choose{N_r}$ is large, the QMCMC-AS method is required to have large $N_{\rm MCMC}$ and $\tau_{\rm stop}$ to find a good subset of antennas \cite{zhang2009receive}.
Accordingly, the complexity of the QMCMC-AS can be unnecessarily high.
Thus, we use the QMCMC-AS method only to provide an approximated optimal performance as a benchmark.

\begin{algorithm}[t!]
\label{algo:QMCMC}
 {\bf Initialization}: Set original distribution $\pi(\pmb \omega)$ as \eqref{eq:originalPDF} and proposal distribution $q(\pmb \omega;\bp)$ as \eqref{eq:proposalPDF} without normalizing factors. 
 Set $\pmb\omega(0)$ as selected antennas from Algorithm \ref{algo:QFAS}, and $\hat{\pmb\omega}_C^* = \pmb\omega(0)$.
 Set $p_j^{(0)} = 1/2$, $\forall j$. \\
 \For{$t = 1:\tau_{\rm stop}$}{
 Run the MIS to draw samples $\{\pmb \omega{(i)}\}_{i=1}^{N_{\rm MCMC}}$ with $p_{\rm accept}( \pi, q)$ \\ 
 If $|\pmb \omega{(i)}| > N_r$, keep only first $N_r$ entries with largest $p_j^{(k)}$. 
 If $|\pmb \omega{(i)}| < N_r$, randomly select $(N_r- |\pmb \omega{(i)}|)$ more antennas.\\
 Update $p_j^{(t)}$ according to \eqref{eq:p_update}.\\
 If $\pi(\pmb \omega{(i)}) > \pi(\hat{\pmb\omega}_C^*)$, for $i= 1,\dots, N_{\rm MCMC}$, set $\pi(\hat{\pmb\omega}_C^*) = \pi(\pmb \omega{(i)})$.
 }
\Return{\ }{$\hat{\pmb\omega}_C^*$\;}
\caption{Quantization-aware MCMC-Antenna Selection (QMCMC-AS) }
\end{algorithm}

%

\section{Extension to Wideband Channels}
\label{sec:extension}

In this section, we derive the multiuser OFDM system models with quantization error and extend the DL and UL antenna selection problems to the wideband OFDM system.
\subsection{Downlink OFDM Communications}
\label{sec:extension_DL}

Let $N_{\rm sc}$ be the number of subcarriers for the OFDM system and $\bu_n \in \bbC^{N_{\rm MS}}$ be the frequency domain symbol vector of $N_{\rm MS}$ MSs at the $n$th subcarrier after ZF precoding for the selected antennas in $\cT$. 
We consider bulk selection where all subcarriers share a same antenna subset.
Then, $\bu_n \in \bbC^{N_{\rm MS}}$ is given as
\begin{align}
	\nonumber
	\bu_n = \bW_{{\rm BB},n}(\cT)\bP_n^{1/2}\bs_n^{\rm dl}
\end{align} 
where $\bW_{{\rm BB},n}(\cT)\in \bbC^{N_t \times N_{\rm MS}}$ is the ZF precoding matrix, $\bP_n = {\rm diag}\{p_{n,1},\dots,p_{n,N_{\rm MS}}\}$ is the power allocation matrix, and $\bs_n = [s_{n,1},s_{n,2},\dots,s_{n,N_{\rm MS}}]^T$  is the frequency symbol vector for the $n$th subcarrier. 
Let $\bx_n^{\rm dl}$ be the DL OFDM symbol vectors at time $n$.
Assuming equal transmit power allocation $p_{n,u} = p_{\mathcal{T}}$, $\forall n,u$, we stack $\bx_n^{\rm dl}$ for $N_{\rm sc}$ time duration $\ubx = [\bx_1^{{\rm dl}\,T}, \bx_2^{{\rm dl}\,T}, \dots, \bx_{N_{\rm sc}}^{{\rm dl}\,T}]^T \in \bbC^{N_{\rm sc}N_t}$, which is given as
\begin{align}
	\nonumber
	\ubx^{\rm dl} &= (\bW_{\rm DFT}^H\otimes\bI_{N_t})\ubu \\
	\nonumber
	& = \sqrt{p_\mathcal{T}}(\bW_{\rm DFT}^H\otimes\bI_{N_t}){\rm BlkDiag}\{\bW_{{\rm BB},1}(\cT), \bW_{{\rm BB},2}(\cT), \dots, \bW_{{\rm BB},N_{\rm sc}}(\cT)\}\ubs^{\rm dl} \\
	\nonumber  
	&= \sqrt{p_\mathcal{T}}(\bW_{\rm DFT}^H\otimes\bI_{N_t})\ubW_{\rm BB}\ubs^{\rm dl}
\end{align}
where  $\bW_{\rm DFT}$ is the normalized $N_{\rm sc}$-point DFT matrix, $\ubu = [\bu_1^T, \bu_2^T, \dots, \bu_{N_{\rm sc}}^T]^T \in \bbC^{N_{\rm sc}N_t}$, $\ubs^{\rm dl} = [\bs_1^{{\rm dl}\, T}, \bs_2^{{\rm dl}\, T}, \dots, \bs_{N_{\rm sc}}^{{\rm dl}\, T}]^T \in \bbC^{N_{\rm sc}N_{\rm MS}}$, and $\ubW_{\rm BB} = {BlkDiag}\{\bW_{{\rm BB},1}(\cT),  \dots, \bW_{{\rm BB},N_{\rm sc}}(\cT)\}$.

Let the analog received signals of $N_{\rm MS}$ MSs after CP removal at time $n$ be $\br_n^{\rm dl} \in \bbC^{N_{\rm MS}}$. 
We stack the vector of received signals $\br_n^{\rm dl} $ for $N_{\rm sc}$ time duration as 
\begin{align}
	\nonumber
	\ubr^{\rm dl}  &= \ubH_{\mathcal{T}}^{\rm dl} \ubx^{\rm dl}  + \ubn^{\rm dl} \\
	\label{eq:r_dl_zf_ofdm}
	& = \sqrt{p_\mathcal{T}}\ubH_{\mathcal{T}}^{\rm dl}(\bW_{\rm DFT}^H\otimes\bI_{N_t})\ubW_{\rm BB}\ubs^{\rm dl}  + \ubn^{\rm dl}
\end{align}
where $\ubr^{\rm dl}  = [\br_1^{{\rm dl}\,T} , \br_2^{{\rm dl}\,T}, \dots, \br_{N_{\rm sc}}^{{\rm dl}\,T}]^T \in \bbC^{N_{\rm sc}N_{\rm MS}}$, and the DL channel matrix for $N_t$ selected transmit antennas $\ubH_{\mathcal{T}}^{\rm dl} \in \bbC^{N_{\rm sc}N_{\rm MS} \times N_{\rm sc}N_t}$ is given as 
\begin{align}
	\ubH_{\mathcal{T}}^{\rm dl} = {\rm BlkCirc}\big\{\bH_{\mathcal{T},0}^{\rm dl}, {\bf 0}, \cdots, {\bf 0}, \bH_{\mathcal{T},L-1}^{\rm dl}, \cdots,\bH_{\mathcal{T},1}^{\rm dl}\big\}
\end{align}
where $\bH_{\mathcal{T},\ell}^{\rm dl}\in \bbC^{N_{\rm MS}\times N_t}$ is the channel matrix of the selected antennas in $\cT$ for the $(\ell+1)$th channel tap, $L$ is the number of channel taps,
and $\ubn^{\rm dl} = [\bn_1^{{\rm dl}\, T}, \bn_2^{{\rm dl}\, T}, \dots, \bn_{N_{\rm sc}}^{{\rm dl}\, T}]^T \in \bbC^{N_{\rm sc}N_{\rm MS}}$ denotes the vector of the AWGN noise vectors stacked for $N_{\rm sc}$ time duration. 

The received OFDM signals $\ubr^{\rm dl}$ are quantized at the ADCs.
The quantized signal are expressed with the AQNM as \cite{fletcher2007robust} 
\begin{align}
	\nonumber
	\uby^{\rm dl} 
	 = \alpha_b\sqrt{p_\mathcal{T}}\ubH_{\mathcal{T}}^{\rm dl}(\bW_{\rm DFT}^H\otimes\bI_{N_t})\ubW_{\rm BB}\ubs^{\rm dl}  + \alpha_b\ubn^{\rm dl} + \ubq^{\rm dl} 
\end{align} 
where $\ubq^{\rm dl} =[\bq_1^{{\rm dl}\, T}, \bq_2^{{\rm dl}\, T}, \dots, \bq_{N_{\rm sc}}^{{\rm dl}\, T}]^T \in \bbC^{N_{\rm sc}N_{\rm MS}}$ is the additive quantization noise vector and $\ubq^{\rm dl} \sim \cC\cN({\bf 0}, \bR_{\underline{\bf q}^{\rm dl}\underline{\bf q}^{\rm dl}})$.
Finally, the quantized signal is combined through a DFT matrix as
\begin{align}
	\nonumber
	\ubz^{\rm dl} &= (\bW_{\rm DFT} \otimes \bI_{N_{\rm MS}})\uby^{\rm dl} \\
	\nonumber
	 &= \alpha_b\sqrt{p_{\mathcal{T}}}(\bW_{\rm DFT} \otimes \bI_{N_{\rm MS}})\ubH_\mathcal{T}^{\rm dl}(\bW_{\rm DFT}^H\otimes\bI_{N_{t}})\ubW_{\rm BB}\ubs^{\rm dl} + (\bW_{\rm DFT} \otimes \bI_{N_{\rm MS}})(\alpha_b\ubn^{\rm dl} + \ubq^{\rm dl} ) \\ 
	\nonumber
	& = \alpha_b\sqrt{p_{\mathcal{T}}}\ubG_\mathcal{T}^{\rm dl} \ubW_{\rm BB}\ubs^{\rm dl} +  \ubv^{\rm dl}\\
	\nonumber
	& \stackrel{(a)}= \alpha_b\sqrt{p_{\mathcal{T}}}\ubs^{\rm dl} +  \ubv^{\rm dl}.
\end{align}
Here, $\ubG_\mathcal{T}^{\rm dl} = (\bW_{\rm DFT} \otimes \bI_{N_{\rm MS}})\ubH_\mathcal{T}^{\rm dl}(\bW_{\rm DFT}^H\otimes\bI_{N_t}) = {\rm BlkDiag}\{\bG^{\rm dl}_{\mathcal{T},1},\cdots,\bG^{\rm dl}_{\mathcal{T},N_{\rm sc}}\}$ where $\bG^{\rm dl}_{\mathcal{T},n} = \sum_{\ell = 0}^{L-1}\bH_{\mathcal{T},\ell}^{\rm dl} \,e^{-\frac{j2\pi(n-1)\ell}{N_{\rm sc}}}$ is the frequency domain DL channel matrix for subcarrier $n$, and $\ubv^{\rm dl} = (\bW_{\rm DFT} \otimes \bI_{N_{\rm MS}})(\alpha_b\ubn^{\rm dl} + \ubq^{\rm dl} ) = [\bv_1^{{\rm dl}\,T},\cdots,\bv_{N_{\rm sc}}^{{\rm dl}\,T}]^T$.
The equality $(a)$ follows from the ZF precoding $\ubW_{\rm BB} = {\rm BlkDiag}\{\bW_{{\rm BB},1}(\cT),\cdots,\bW_{{\rm BB},N_{\rm sc}}(\cT)\} = \ubG_\mathcal{T}^{{\rm dl}\,H}(\ubG_\mathcal{T}^{\rm dl}\ubG_\mathcal{T}^{{\rm dl}\,H})^{-1}$, i.e., 
\begin{align}
	\nonumber
	\bW_{{\rm BB},n}(\cT) = \bG^{{\rm dl}\,H}_{\mathcal{T},n}(\bG^{\rm dl}_{\mathcal{T},n}\bG^{{\rm dl}\,H}_{\mathcal{T},n})^{-1}.
\end{align}
Under coarse quantization, the received digital signal after DFT for subcarrier $n$ becomes 
\begin{align}
	\nonumber
	\bz_n^{\rm dl} = \alpha_b\sqrt{p_\mathcal{T}}\bs_n^{\rm dl} + \bv_n^{\rm dl}.
\end{align}

Now, we compute the covariance matrix of $\bv_n^{\rm dl}$. 
Let $\bW_{\rm MS} = (\bW_{\rm DFT} \otimes \bI_{N_{\rm MS}})$ and $\bW_{\rm BS} = (\bW_{\rm DFT} \otimes \bI_{N_{t}})$. 
Then, the covariance matrix of $\bv_n^{\rm dl}$ is expressed as
\begin{align}
	\nonumber
	\bR_{\bv_n^{\rm dl}\bv_n^{\rm dl}} 
	&= \alpha_b^2\bW_{{\rm MS},n} \bbE\big[\ubn^{\rm dl} \ubn^{{\rm dl}\,H} \big]\bW_{{\rm MS},n}^{H}+ \bW_{{\rm MS},n}\bbE\big[\,\ubq^{\rm dl}\ubq^{{\rm dl}\,H}\big]\bW_{{\rm MS},n}^{H}\\
	\nonumber
	&= \alpha_b^2\bI_{N_{\rm MS}} + \bW_{{\rm MS},n} \bR_{\underline{\bf q}^{{\rm dl}}\underline{\bf q}^{{\rm dl}}}\bW_{{\rm MS},n}^{H}
\end{align}
where $\bW_{{\rm MS},n} = ([\bW_{\rm DFT}]_{n,:}\otimes \bI_{N_{\rm MS}})$, and $\bR_{\underline{\bf q}^{{\rm dl}}\underline{\bf q}^{{\rm dl}}} = \bbE\big[\,\ubq^{\rm dl}\ubq^{{\rm dl}\,H}\big]$ is the covariance matrix of $\ubq^{\rm dl}$.
To derive $\bR_{\underline{\bf q}^{{\rm dl}}\underline{\bf q}^{{\rm dl}}}$, we first simplify the precoding matrix $\ubW_{\rm BB}$ as follows:
\begin{align}
	\nonumber
	\ubW_{\rm BB} &= \ubG_\mathcal{T}^{{\rm dl}\,H}(\ubG_\mathcal{T}^{\rm dl}\ubG_\mathcal{T}^{{\rm dl}\,H})^{-1}\\
	\nonumber
	& \stackrel{(a)}= \bW_{\rm BS}\ubH_\mathcal{T}^{{\rm dl}\,H}\bW_{\rm MS}^H\left(\bW_{\rm MS}\ubH_\mathcal{T}^{\rm dl}\bW_{\rm BS}^H\bW_{\rm BS}\ubH_\mathcal{T}^{{\rm dl}\,H}\bW_{\rm MS}^H\right)^{-1}\\
	\label{eq:WBB_simple}
	& \stackrel{(b)}= \bW_{\rm BS}\ubH_\mathcal{T}^{{\rm dl}\,H}\left(\ubH_\mathcal{T}^{\rm dl}\ubH_\mathcal{T}^{{\rm dl}\,H}\right)^{-1}\bW_{\rm MS}^{-1}
\end{align}
where $(a)$ comes from the definition of $\ubG_\mathcal{T}^{\rm dl} = \bW_{\rm MS}\ubH_\mathcal{T}^{\rm dl}\bW_{\rm BS}^H$ and $(b)$ follows from the fact that $\bW_{\rm MS}$, $\bW_{\rm BS}$, and $\ubH_\mathcal{T}^{\rm dl}\ubH_\mathcal{T}^{{\rm dl}\,H}$ are invertible.
Then, the covariance matrix of $\ubq^{\rm dl}$ becomes \cite{fletcher2007robust,fan2015uplink}
\begin{align}
	\nonumber
	\bR_{\underline{\bf q}^{{\rm dl}}\underline{\bf q}^{{\rm dl}}} &= \alpha_b(1-\alpha_b){\rm diag}\left\{\bbE[\ubr^{\rm dl}\ubr^{{\rm dl}\,H}]\right\} \\
	\nonumber
	& = \alpha_b(1-\alpha_b){\rm diag}\big\{p_\mathcal{T}\ubH_\mathcal{T}^{\rm dl}\bW_{\rm BS}^H\ubW_{\rm BB}\ubW_{\rm BB}^H\bW_{\rm BS}\ubH_\mathcal{T}^{{\rm dl}\,H} + \bI_{N_{\rm sc}N_{\rm MS}}\big\}\\
	\label{eq:Rqq_dl_OFDM}
	& \stackrel{(a)}= 
	 \alpha_b(1-\alpha_b)(p_\mathcal{T}+1)\bI_{N_{\rm sc}N_{\rm MS}}
\end{align}
where $(a)$ follows from \eqref{eq:WBB_simple}. 
Finally, using \eqref{eq:Rqq_dl_OFDM}, the covariance matrix $\bR_{\bv_n^{\rm dl}\bv_n^{\rm dl}}$ becomes $\bR_{\bv_n^{\rm dl}\bv_n^{\rm dl}} = (\alpha_b+\alpha_b(1-\alpha_b)p_\mathcal{T})\bI_{N_{\rm MS}}$.
Accordingly, the SINR of user $u$ for $n$th subcarrier is given as
\begin{align}
	\label{eq:SINR_dl_OFDM}
	{\rm SINR}_{u,n}(\cT) = \frac{\alpha_b p_\mathcal{T}}{1+(1-\alpha_b)p_\mathcal{T}}.	
\end{align}

Using \eqref{eq:SINR_dl_OFDM}, we formulate the transmit antenna selection problem for the OFDM system as
\begin{align}
	\nonumber
	\cP 3:\quad\quad \cT^\star_{\rm ofdm} = \argmax_{\mathcal{T}\subseteq \mathcal{S} :|\mathcal{T}| = N_t\geq N_{\rm MS}}\cR^{\rm dl,ofdm}(\cT)
\end{align}
where $\cR^{\rm dl,ofdm}(\cT) = \frac{1}{N_{\rm sc}} \sum_{n=1}^{N_{\rm sc}}\sum_{u=1}^{N_{\rm MS}}\log_2\big(1+{\rm SINR}_{u,n}(\cT)\big)$ is the average sum rate.
From \eqref{eq:SINR_dl_OFDM}, it can be shown that the achievable rate is equal for all $u$ and $n$. 
Consequently, maximizing the sum rate is equivalent to maximizing the SINR in \eqref{eq:SINR_dl_OFDM}, and we need to select transmit antennas that maximize the transmit power $p_\mathcal{T}$.
We consider that the total transmit power is constrained by $P$ as ${\rm tr}\{\bbE[\ubx^{\rm dl}\ubx^{{\rm dl}\, H}]\}\leq P$. 
Assuming equal power allocation for each user and subcarrier, we have ${\rm tr}\big\{\bbE[\ubx^{\rm dl}\ubx^{{\rm dl}\, H}]\big\} = p_\mathcal{T}{\rm tr}\big\{\bW_{\rm BS}^H\ubW_{\rm BB}\ubW_{\rm BB}^H\bW_{\rm BS}\big\}
 = p_\mathcal{T} {\rm tr}\big\{(\ubH_\mathcal{T}^{\rm dl}\ubH_\mathcal{T}^{{\rm dl}\, H})^{-1}\big\}$
and thus, the power allocation $p_\mathcal{T}$ with maximum transmit power is given as
\begin{align}
	\label{eq:pt_dl_ofdm}
	p_\mathcal{T} = \frac{P}{{\rm tr}\big\{(\ubH_\mathcal{T}^{\rm dl}\ubH_\mathcal{T}^{{\rm dl}\, H})^{-1}\big\}}.
\end{align} 
\begin{remark}
The transmit power in \eqref{eq:pt_dl_ofdm} shows that the transmit antenna selection problem for DL OFDM communications in low-resolution ADC systems with ZF precoding and equal power allocation is equivalent to that in high-resolution ADC systems.
\end{remark}
Accordingly, any state-of-the-art transmit antenna selection methods for high-resolution ADC OFDM systems with ZF-precoding can be employed for low-resolution ADC OFDM systems, which was also true for narrowband communications as shown in Section~\ref{sec:DL_txantenna}.
In addition, we note that the analysis derived in Section~\ref{subsec:DL_analysis} also holds for the DL OFDM systems.
\begin{corollary}
	\label{cor:monotonic_dl_ofdm}
	For the multiuser DL OFDM system with ZF precoding and equal power distribution in \eqref{eq:r_dl_zf_ofdm}, the maximum achievable sum rate of MSs with low-resolution ADCs is monotonically increasing with the number of selected transmit antennas:
	\begin{align}
		\nonumber
		\cR^{\rm dl,ofdm}(\cT_{\rm opt1}) < \cR^{\rm dl,ofdm}(\cT_{\rm opt2})
	\end{align}
	where $\cT_{\rm opt1}$ and $\cT_{\rm opt2}$ are the optimal antenna subsets with $|\cT_{\rm opt1}| < |\cT_{\rm opt2}|$.
\end{corollary}
\begin{proof}
	We replace $\bH_\mathcal{T}^{\rm dl}$ in the proof of Theorem~\ref{thm:monotonic} with $\ubH_\mathcal{T}^{\rm dl}$ and follow the same proof.
\end{proof}
According to Corollary~\ref{cor:monotonic_dl_ofdm}, we need to use as many antennas at the BS for DL OFDM systems with ZF-precoding to maximize the achievable rate even with quantization error at the MSs.

\subsection{Uplink ODFM Communications}
\label{sec:extension_UL}

Similarly to the DL OFDM system model with low-resolution ADCs derived in the previous section, the UL ODFM system with low-resolution ADCs  can be modeled as follows \cite{prasad2019optimizing}. 
Let $\bx_n^{\rm ul} \in \bbC^{N_{\rm MS}}$ be a vector of the OFDM symbols of $N_{\rm MS}$ MSs at time $n$. 
Let $\ubx^{\rm ul} = [\bx_1^{{\rm ul}\,T}, \bx_2^{{\rm ul}\,T}, \dots, \bx_{N_{\rm sc}}^{{\rm ul}\,T}]^T \in \bbC^{N_{\rm sc}N_{\rm MS}}$, which is given as
\begin{align}
	\nonumber
	\ubx^{\rm ul} = \sqrt{\rho}(\bW_{\rm DFT}^H\otimes\bI_{N_{\rm MS}})\ubs^{\rm ul}
\end{align}
where $\ubs^{\rm ul} = [\bs_1^{{\rm ul}\,T}, \bs_2^{{\rm ul}\,T}, \dots, \bs_{N_{\rm sc}}^{{\rm ul}\,T}]^T \in \bbC^{N_{\rm sc}N_{\rm MS}}$ and $\bs_n^{{\rm ul}} = [s_{n,1}^{\rm ul},s_{n,2}^{\rm ul},\dots,s_{n,N_{\rm MS}}^{\rm ul}]^T$.

Let the analog received signals at the BS with $N_r$ selected antennas in $\cK$ after CP removal at time $n$ be $\br_n^{\rm ul} \in \bbC^{N_r}$. 
We stack the vector of received signals $\br_n^{\rm ul}$ for $N_{\rm sc}$ time duration as
\begin{align}
	\nonumber
	\ubr^{\rm ul} &= \ubH_{\mathcal{K}}^{\rm ul}\ubx^{\rm ul} + \ubn^{\rm ul}\\
	\nonumber
	&= \sqrt{\rho}\ubH_{\mathcal{K}}^{\rm ul}(\bW_{\rm DFT}^H\otimes\bI_{N_{\rm MS}})\ubs^{\rm ul} + \ubn^{\rm ul}
\end{align}
where $\ubr^{\rm ul} = [\br_1^{{\rm ul}\,T}, \br_2^{{\rm ul}\,T}, \dots, \br_{N_{\rm sc}}^{{\rm ul}\,T}]^T \in \bbC^{N_{\rm sc}N_r}$, and the UL channel matrix in the time domain for $N_r$ selected antennas $\ubH_{\mathcal{K}}^{\rm ul} \in \bbC^{N_{\rm sc}N_r \times N_{\rm sc}N_{\rm MS}}$ is given as 
\begin{align}
	\nonumber
	\ubH_{\mathcal{K}}^{\rm ul} = {\rm BlkCirc}\big\{\bH_{\mathcal{K},0}^{\rm ul}, {\bf 0}, \cdots, {\bf 0}, \bH_{\mathcal{K},L-1}^{\rm ul}, \cdots,\bH_{\mathcal{K},1}^{\rm ul}\big\}
\end{align}
where $\bH_{\mathcal{K},\ell}^{\rm ul}$ is the UL channel matrix of the selected antennas for the $(\ell+1)$th channel tap, $L$ is the number of channel taps, and $\ubn^{\rm ul} = [\bn_1^{{\rm ul}\,T}, \bn_2^{{\rm ul}\,T}, \dots, \bn_{N_{\rm sc}}^{{\rm ul}\,T}]^T \in \bbC^{N_{\rm sc}N_r}$ denotes the vector of the AWGN noise vectors.

After quantization, the quantized OFDM signals are expressed by adopting the AQNM as \cite{fletcher2007robust} 
\begin{align}
	\nonumber
	\uby^{\rm ul} 
	= \alpha_b\sqrt{\rho}\ubH_\mathcal{K}^{\rm ul}(\bW_{\rm DFT}^H\otimes\bI_{N_{\rm MS}})\ubs^{\rm ul} + \alpha_b\ubn^{\rm ul} + \ubq^{\rm ul}
\end{align} 
where $\ubq^{\rm ul} =[\bq_1^{{\rm ul}\,T}, \bq_2^{{\rm ul}\,T}, \dots, \bq_{N_{\rm sc}}^{{\rm ul}\,T}]^T \in \bbC^{N_{\rm sc}N_r}$ is the additive quantization noise vector and $\bq_n^{\rm ul} \sim \cC\cN({\bf 0}, \bR_{\bq_n^{\rm ul}\bq_n^{\rm ul}})$.
The covariance matrix $\bR_{\bq_n^{\rm ul}\bq_n^{\rm ul}}$ is derived as \cite{fletcher2007robust}
\begin{align}
	\nonumber
	\bR_{\bq_n^{\rm ul}\bq_n^{\rm ul}} &=  \alpha_b(1-\alpha_b){\rm diag}\{\bbE[\br_n^{\rm ul}\br_n^{{\rm ul}\,H}]\}\\
	\label{eq:Rqq_uplink_wide}
	& = \alpha_b(1-\alpha_b){\rm diag}\big\{ \rho\bB_\mathcal{K}\bB_\mathcal{K}^H+\bI_{N_r}\big\}
\end{align}
where $\bB_\mathcal{K} = [\bH_{\mathcal{K},0}^{\rm ul}, {\bf 0}, \cdots, {\bf 0}, \bH_{\mathcal{K},L-1}^{\rm ul}, \cdots,\bH_{\mathcal{K},1}^{\rm ul}]$.
We note that $\bR_{\bq_n^{\rm ul}\bq_n^{\rm ul}}= \bR_{\bq_{m}^{\rm ul}\bq_{m}^{\rm ul}}, \forall n \neq m$, i.e., $\bR_{\bq_n^{\rm ul}\bq_n^{\rm ul}}$ is independent to subcarriers.
Finally, $\uby^{\rm ul}$ is combined through a DFT matrix as
\begin{align}
	\nonumber
	\ubz^{\rm ul} &= (\bW_{\rm DFT} \otimes \bI_{N_r})\uby^{\rm ul} \\
	\nonumber
	& = \alpha_b\sqrt{\rho}(\bW_{\rm DFT} \otimes \bI_{N_r})\ubH_\mathcal{K}^{\rm ul}(\bW_{\rm DFT}^H\otimes\bI_{N_{\rm MS}})\ubs^{\rm ul} + (\bW_{\rm DFT} \otimes \bI_{N_r})(\alpha_b\ubn^{\rm ul} + \ubq^{\rm ul} ) \\ 
	\nonumber
	& = \alpha_b\sqrt{\rho}\ubG_\mathcal{K}^{\rm ul} \ubs^{\rm ul} +  \ubv^{\rm ul}
\end{align}
where $\ubG_\mathcal{K}^{\rm ul} = (\bW_{\rm DFT} \otimes \bI_{N_r})\ubH_\mathcal{K}^{\rm ul}(\bW_{\rm DFT}^H\otimes\bI_{N_{\rm MS}}) = {\rm BlkDiag}\{\bG_{\mathcal{K},1}^{\rm ul},\cdots,\bG_{\mathcal{K},N_{\rm sc}}^{\rm ul}\}$, $\bG_{\mathcal{K},n}^{\rm ul} = \sum_{\ell = 0}^{L-1}\bH_{\mathcal{K},\ell}^{\rm ul} e^{-\frac{j2\pi(n-1)\ell}{N_{\rm sc}}}$, and $\ubv^{\rm ul} = (\bW_{\rm DFT} \otimes \bI_{N_r})(\alpha_b\ubn^{\rm ul} + \ubq^{\rm ul} ) = [\bv_1^{{\rm ul}\,T},\cdots,\bv_{N_{\rm sc}}^{{\rm ul}\,T}]^T$.
Accordingly, under coarse quantization, the received digital signal after DFT for subcarrier $n$ becomes 
\begin{align}
	\label{eq:zn_ul_ofdm}
	\bz_n^{\rm ul} = \alpha_b\sqrt{\rho}\bG_{\mathcal{K},n}^{\rm ul} \bs_n^{\rm ul} + \bv_n^{\rm ul}.
\end{align}
The covariance matrix of $\bv_n^{\rm ul}$ is derived as $\bR_{\bv_n^{\rm ul}\bv_n^{\rm ul}}= \alpha_b^2\bI_{N_r} + \bR_{\bq_n^{\rm ul}\bq_n^{\rm ul}}$
where $\bR_{\bq_n^{\rm ul}\bq_n^{\rm ul}}$ is defined in \eqref{eq:Rqq_uplink_wide}.
Using \eqref{eq:zn_ul_ofdm}, the UL capacity for subcarrier $n$ is derived as
\begin{align}
	\label{eq:Rn_uplink_wide}
	\cR_n^{\rm ul}(\cK) = \log_2\Big|\bI_{N_r} + \rho\alpha_b^2(\alpha^2_b \bI_{N_r} + \bR_{\bq_n^{\rm ul}\bq_n^{\rm ul}})^{-1}\bG_{\mathcal{K},n}^{\rm ul}\bG_{\mathcal{K},n}^{{\rm ul}\,H}\Big|.
\end{align}
Note that the capacity of the wideband OFDM system for each subcarrier in \eqref{eq:Rn_uplink_wide} shows similar structure as that of the narrowband system in \eqref{eq:capacity}.

Since all subcarriers share a same subset of antennas, i.e., $\cK$ is same for all subcarriers, the maximization cannot be applied to each subcarrier separately. 
Accordingly, we need to find the best subset of antennas $\cK$ for the entire subcarriers, and the receive antenna selection problem for the wideband UL OFDM system is formulated as
\begin{align}
	\label{eq:problem_ofdm_uplink}
	\cP 4:\quad\quad\cK^{\star}_{\rm ofdm} = \argmax_{\mathcal{K}\subseteq \mathcal{S}:|\mathcal{K}|= N_r\geq N_{\rm MS}} \sum_{n=1}^{N_{\rm sc}} \cR_n^{\rm ul}(\cK).
\end{align}
To solve \eqref{eq:problem_ofdm_uplink}, we extend the greedy approach for the narrowband communications in Section~\ref{sec:UL_rxantenna}.
We also show that the MCMC approach can be naturally adopted with modification.

Similarly to \eqref{eq:greedy-max},  let $\bG^{\rm ul}_{\mathcal{K}_t\cup \{j\},n}$ be the channel matrix of $t$ selected antennas during the first $t$ greedy selections and a candidate antenna $j\in \mathcal{S}\setminus \mathcal{K}_t$ at the next selection. 
Then, the greedy maximization problem is formulated as
\begin{align}
	\label{eq:problem_greedy_ul_ofdm}
	J = \argmax_{j \in \mathcal{S}\setminus \mathcal{K}_t} \sum_{n=1}^{N_{\rm sc}}\cR_n^{\rm ul}\left({\cK}_t\cup \{j\}\right).
\end{align}
Now, we decompose \eqref{eq:Rn_uplink_wide}.
Let $\bar{\bD}_{\mathcal{K}_t\cup\{j\}}  = \bI_{t+1} + \rho(1-\alpha_b){\rm diag}\{\bB_{\mathcal{K}_t\cup \{j\}}\bB_{\mathcal{K}_t\cup \{j\}}^H\}$. 
At the $(t+1)$th selection stage, we have
\begin{align}
	\nonumber
    \cR_n^{\rm ul}(\cK_t\cup \{j\}) 
	& = \log_2 \Big|{\bf I}_{N_{\rm MS}} + \rho\alpha_b{\bf G}^{{\rm ul}\,H}_{{\mathcal{K}_t\cup \{j\}},n} \bar{\bf D}^{-1}_{{\mathcal{K}_t\cup \{j\}}}{\bf G}^{\rm ul}_{{\mathcal{K}_t\cup \{j\}},n} \Big|\\
	\nonumber
	& = \log_2 \biggl|{\bf I}_{N_{\rm MS}} + \rho\alpha_b\Bigl({\bf G}^{{\rm ul}\,H}_{\mathcal{K}_t,n} \bar{\bf D}^{-1}_{\mathcal{K}_t}{\bf G}^{\rm ul}_{\mathcal{K}_t,n}\! +\! \frac{1}{\bar{d}_{j}}{\bf f}_{n,j}{\bf f}^H_{n,j}\Bigr)\biggr|\\
	\label{eq:rate_reduced_ul_ofdm}
	&= \cR^{\rm ul}_n(\cK_t) + \log_2\biggl(1\!+\!\frac{\rho \alpha_b}{d_j}c_{n,t}(j) \biggr)
\end{align}
where  ${\bf f}^H_{n,j}$ is $j$th row of ${\bf G}_n^{\rm ul}$, $\bar{d}_{j}$ is the corresponding diagonal entry of $\bar{\bf D}_{{\mathcal{K}_t\cup \{j\}}}$, and $c_{n,t}(j)$ is  
\begin{align}
	\label{eq:capacity gain_ofdm}
	c_{n,t}(j) ={\bf f}^H_{n,j} \Bigl({\bf I}_{N_{\rm MS}}\! +\! \rho\alpha_b{\bf G}^{{\rm ul}\,H}_{\mathcal{K}_t,n} \bar{\bf D}^{-1}_{\mathcal{K}_t}{\bf G}^{\rm ul}_{\mathcal{K}_t,n}\Bigr)^{-1}{\bf f}_{n,j}.
\end{align}
With \eqref{eq:rate_reduced_ul_ofdm}, the greedy maximization problem  in \eqref{eq:problem_greedy_ul_ofdm} reduces to
\begin{align}
	\label{eq:problem3_ul_ofdm}
	J = \argmax_{j\in \mathcal{S}\setminus \mathcal{K}_t} \sum_{n=1}^{N_{\rm sc}} \log_2\left(1+\frac{\rho\alpha_b}{d_j}c_{n,t}(j)\right).
\end{align}
Therefore, a greedy algorithm that is similar to Algorithm \ref{algo:QFAS} can be used for \eqref{eq:problem3_ul_ofdm}.
In addition, let $\bQ_{n,t}=({\bf I}_{N_{\rm MS}}\! +\! \rho\alpha_b{\bf G}^{{\rm ul}\,H}_{\mathcal{K}_t,n} \bar{\bf D}^{-1}_{\mathcal{K}_t}{\bf G}^{\rm ul}_{\mathcal{K}_t,n})^{-1}$. 
Then, $c_{n,t}(j)$ in \eqref{eq:capacity gain_ofdm} can also be updated without matrix inversion for each subcarrier as shown in Algorithm~\ref{algo:QFAS}.
Accordingly, the complexity of the proposed QFAS algorithm for the UL OFDM system becomes $\cO(N_{\rm sc} N_r N_{\rm MS} N_{\rm BS})$.
\begin{corollary}\label{cor:submodular_wide}
	The capacity of the QFAS method for the UL OFDM system is lower bounded by
	\begin{align}
		\label{eq:lowerbound_ofdm}
		\sum_{n=1}^{N_{\rm sc}}\cR^{\rm ul}_n(\cK_{\rm qfas}) \geq \left(1-\frac{1}{e}\right)\sum_{n=1}^{N_{\rm sc}}\cR^{\rm ul}_n(\cK^\star_{\rm ofdm})
	\end{align}
	where $\cK^\star_{\rm ofdm}$ is the optimal subset of receive antennas defined in \eqref{eq:problem_ofdm_uplink}.
\end{corollary}
\begin{proof}
	The class of submodular functions is closed under nonnegative linear combinations, and we showed that the capacity with the quantization error is submodular in the proof of Corollary~\ref{cor:submodular_flat}.
	Consequently, the sum capacity for all carrier frequencies in \eqref{eq:problem_ofdm_uplink} is also submodular. 
	Since the proposed QFAS for the wideband OFDM system solves \eqref{eq:problem3_ul_ofdm}, which is equivalent to the greedy maximization in \eqref{eq:problem_greedy_ul_ofdm}, from Theorem~\ref{thm:submodular}, we derive \eqref{eq:lowerbound_ofdm}.
\end{proof}

To find an approximated optimal solution, we can also use the adaptive MCMC approach described in \ref{subsec:mcmc}. 
To this end, the original PDF $\pi(\pmb \omega)$ needs to be modified as 
\begin{align}
	\label{eq:originalPDF_ofdm}
	\pi(\pmb \omega) \triangleq  \exp\left(\frac{1}{\tau}{\sum_{n=1}^{N_{\rm sc}}\mathcal{R}^{\rm ul}_n(\pmb \omega)}\right)/\Gamma_{\rm ofdm}
\end{align}
where $\tau$ is a rate constant and $\Gamma_{\rm ofdm}$ is a normalizing factor for the PDF.
Then, the adaptive MCMC-based antenna selection method for the OFDM system can be performed similarly to the QMCMC-AS method in \ref{subsec:mcmc}.
The complexity of the QMCMC-AS method for the OFDM system is $\cO(N_{\rm sc}N_rN_{\rm MS}^2N_{\rm MCMC}\tau_{\rm stop})$.
\section{Simulation Results}
\label{sec:simulation}

In this section, we validate the theoretical results and proposed methods.
We assume Rayleigh channels with a zero mean and unit variance for small scale fading.
For a large scale fading, we adopt the log-distance pathloss model \cite{erceg1999empirically}.
We consider randomly distributed MSs over a single cell with radius of $1km$.
We assume the minimum distance between the BS and MSs to be $100m$.
Considering a $2.4$ GHz carrier frequency with $10$ MHz bandwidth, we use $8.7$ dB lognormal shadowing variance and $12$ dB noise figure at receivers.

\subsection{Downlink Transmit Antenna Selection}
\label{subsec:sim_dl_ofdm}

\begin{figure}[!t]
\centering
$\begin{array}{c c}
{\resizebox{0.48\columnwidth}{!}
{\includegraphics{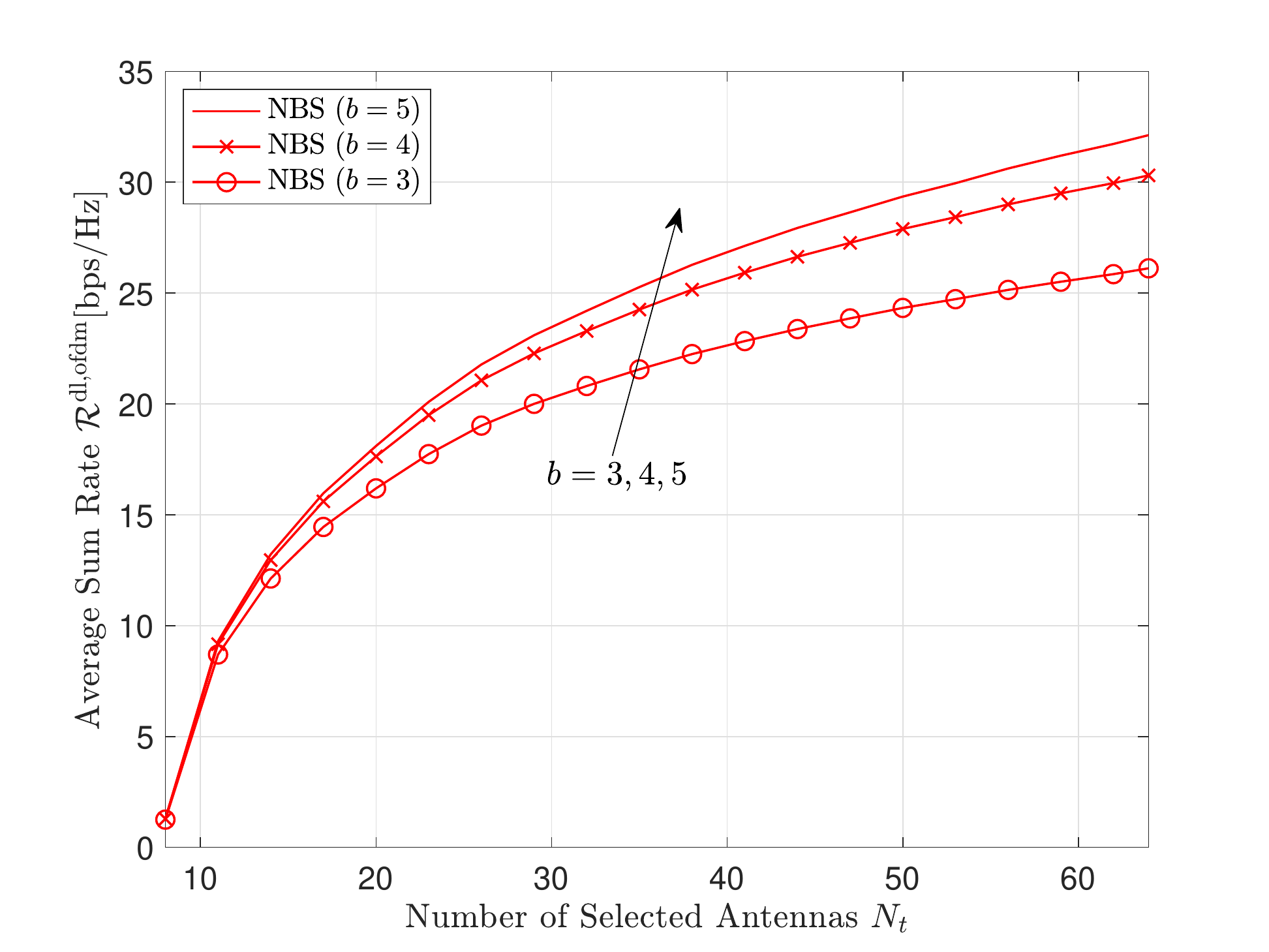}}
} &
{\resizebox{0.48\columnwidth}{!}
{\includegraphics{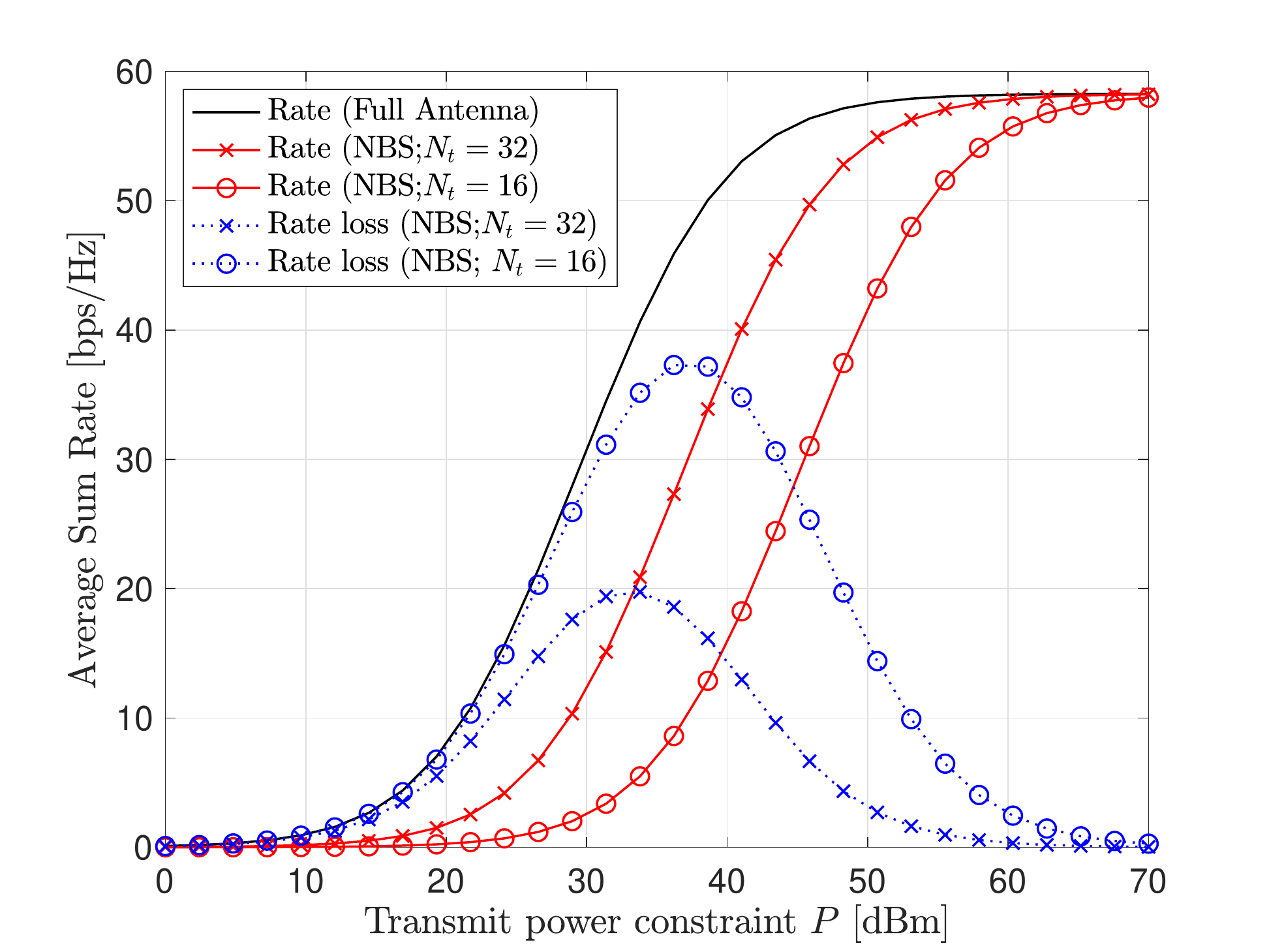}}
}\\
\mbox{\small (a)}& \mbox{\small (b)}
\end{array}$
\caption{
Average sum rate $\cR^{\rm dl,ofdm}$ (a) with respect to the number of selected antennas $N_t$ for $N_{\rm BS} = 64$ BS antennas, $N_{\rm MS} = 8$ MSs, $P = 30$ dBm total power constraint, and $b \in  \{3, 4, 5\}$ ADC bits, and (b) with respect to the total transmit power constraint $P$ for $N_{\rm BS} =128$ BS antennas, $N_{\rm MS} = 12$ MSs, $N_t =16$ selected antennas, and $b = 3$ ADC bits.} 
\label{fig:rate_dl_ofdm}
\end{figure}

We consider the DL ODFM system with $N_{\rm sc} = 64$ subcarriers for channels with $L = 4$ taps.
To validate the analysis, we use the norm-based selection (NBS) method in simulations, which selects antennas in the order of channel norm that corresponds to each antenna \cite{liu2009low, zhang2009receive}.
Note that the NBS method always provides $\cT_1 \subseteq \cT_2$ when $|\cT_1| \leq |\cT_2|$ for the same channel.
In Fig.~\ref{fig:rate_dl_ofdm}(a), the average sum rate increases with the number of selected antennas, which validates the derived Theorem~\ref{thm:monotonic} and Corollary~\ref{cor:monotonic_dl_ofdm}. 
Fig.~\ref{fig:rate_dl_ofdm}(b) shows the average sum rate versus the total power constraint $P$. 
Unlike the high-resolution ADC systems, there exists a point $P_{D}^{\rm max}$ for the maximum rate loss from not using all antennas, and the rate loss decreases after the point $P_{D}^{\rm max}$ in \eqref{eq:Pmax} for the OFDM channel $\ubH^{\rm dl}$.
Theoretical $P_{D}^{\rm max}$ for the NBS method with $N_t = 32$ and $N_t = 16$ are $33.1351$ dBm and $37.2850$, respectively.
In addition, the theoretical maximum rate loss in \eqref{eq:RLmax} for the OFDM channel $\ubH^{\rm dl}$ with $N_t = 32$ and $N_t = 16$ are $19.8034$ bps/Hz and $37.5282$ bps/Hz, respectively, which also corresponds to the simulation results.

\subsection{Uplink Receive Antenna Selection}
\label{subsec:sim_ul}

\begin{figure}[!t]
\centering
$\begin{array}{c c}
{\resizebox{0.48\columnwidth}{!}
{\includegraphics{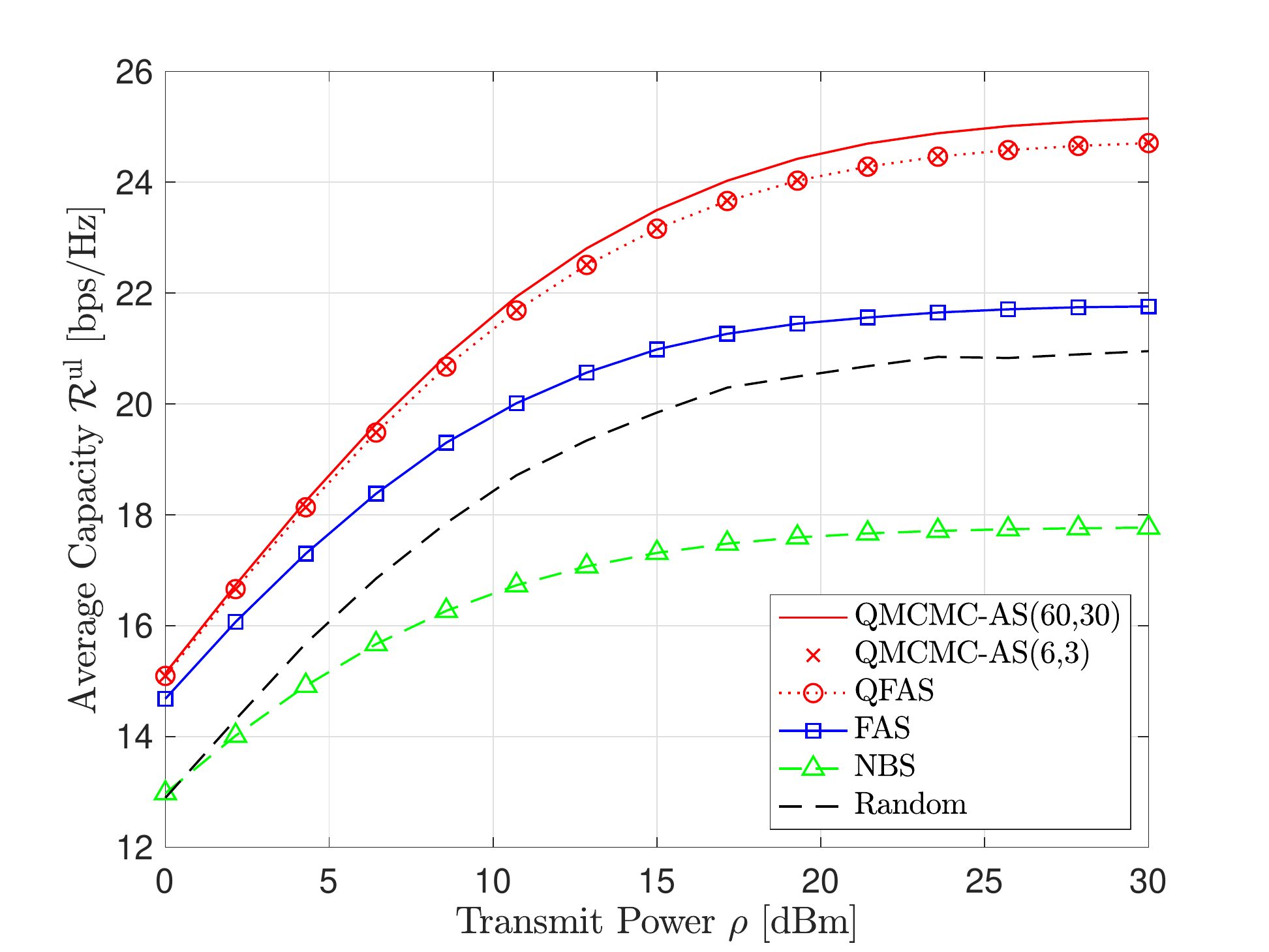}}
} &
{\resizebox{0.48\columnwidth}{!}
{\includegraphics{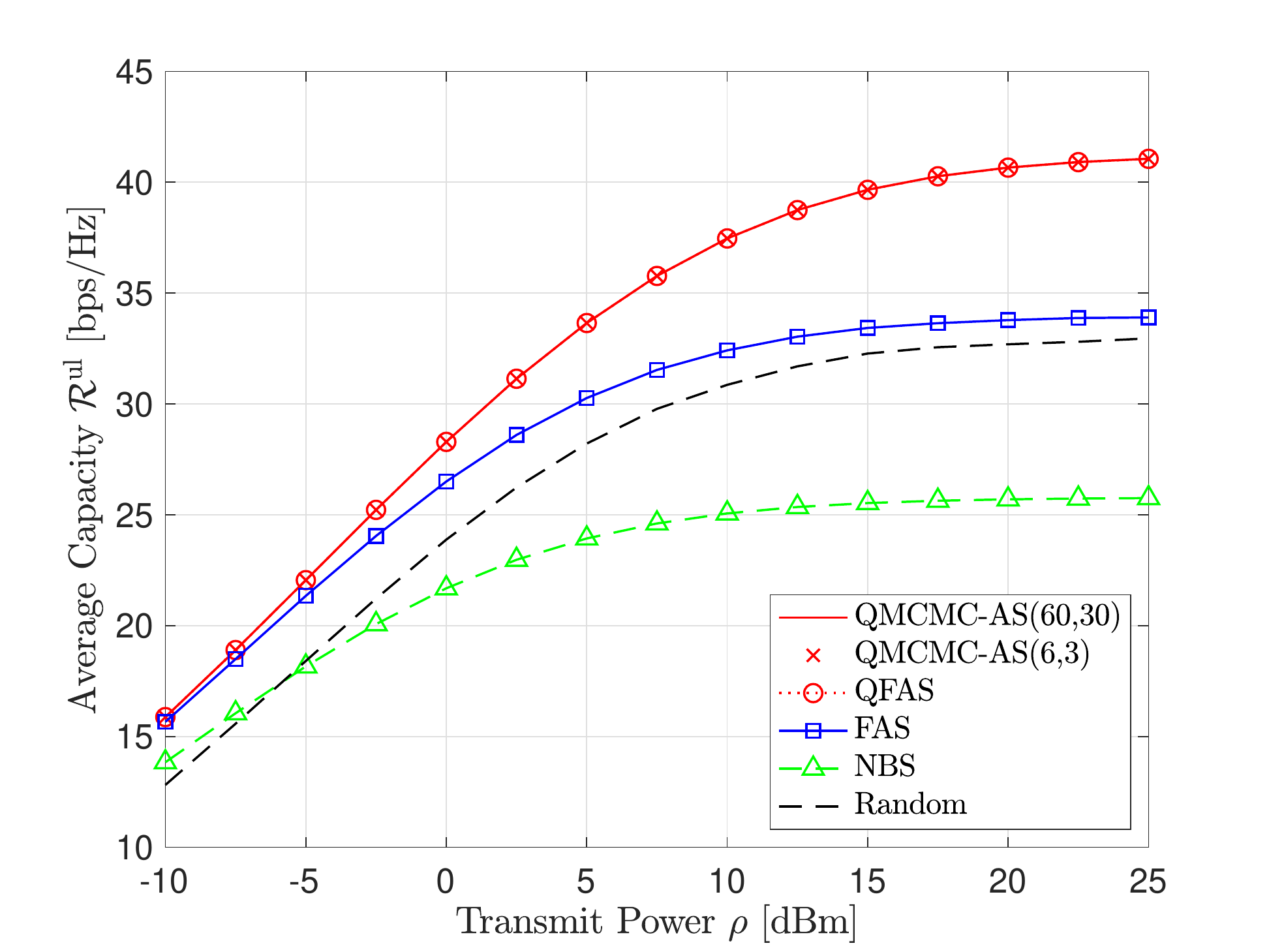}}
}\\
\mbox{\small (a)}& \mbox{\small (b)}
\end{array}$
\caption{
Average capacity $\mathcal{R}^{\rm ul}$ with respect to transmit power $\rho$ for (a) $N_{\rm BS} = 32$ BS antennas, $N_{\rm MS} = 8$ MSs, $N_r =8$ selected antennas, and $b = 3$ quantization bits, and for (b) $N_{\rm BS} =128$ BS antennas, $N_{\rm MS} = 12$ MSs, $N_r =16$ selected antennas, and $b = 3$ ADC bits.} 
\label{fig:capacity_power}
\end{figure}

We evaluate the proposed algorithms for the UL antenna selection---QFAS and QMCMC-AS methods.
We also simulate the NBS method \cite{liu2009low, zhang2009receive} and the fast antenna selection (FAS) algorithm in \cite{gharavi2004fast}, which shows a comparable performance to the optimal selection under perfect quantization.
Although the NBS method presents low performance improvement, because of its low complexity $\cO(N_{\rm MS}N_r)$, it is considered as a reasonable antenna selection method for high-resolution ADC systems \cite{zhang2009receive}.
A random selection is simulated to offer a reference performance.
\subsubsection{Narrowband Communications}

\begin{figure}[!t]\centering
	\includegraphics[scale = 0.43]{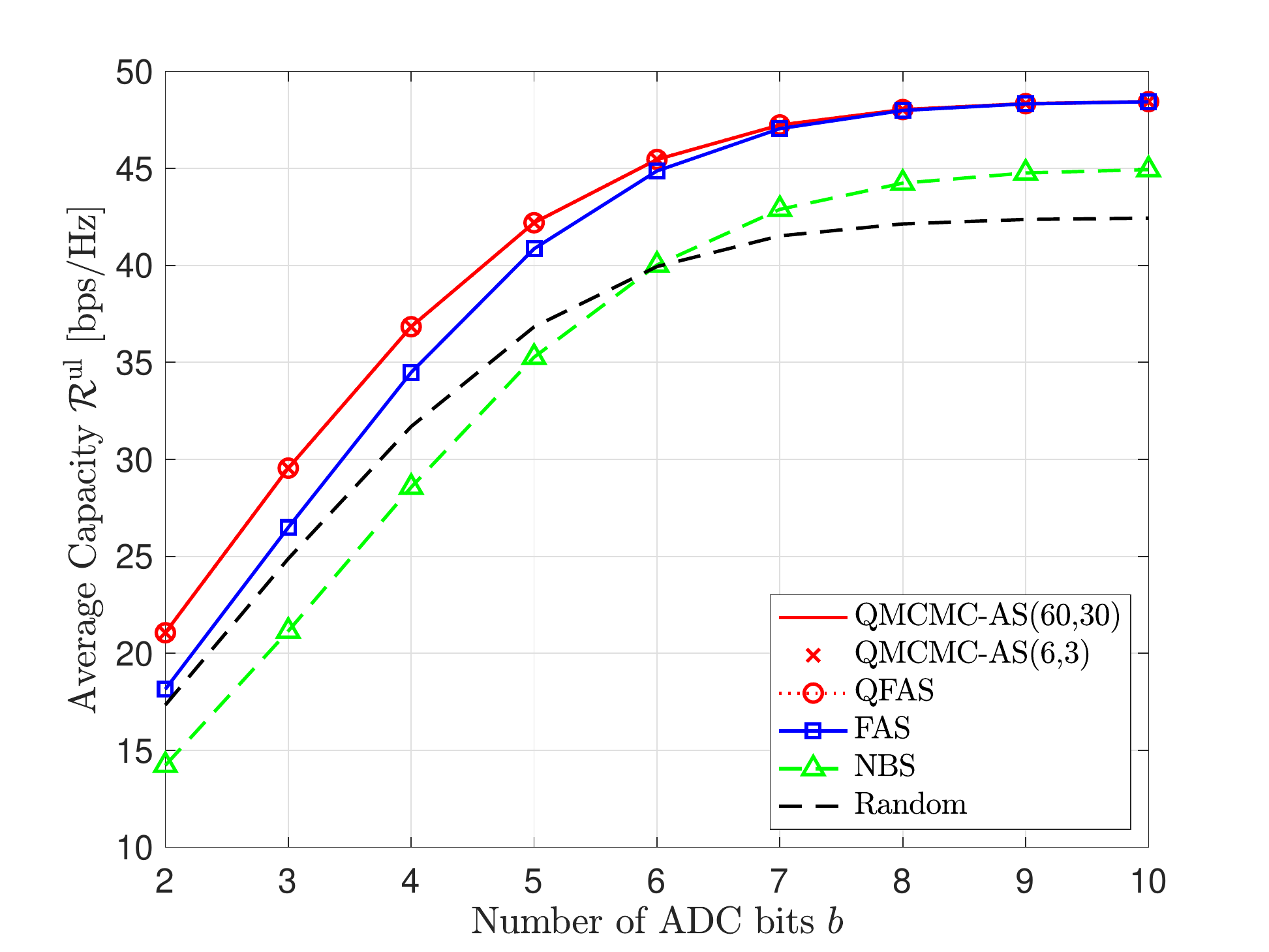}
	\caption{Average capacity $\mathcal{R}^{\rm ul}$ with respect to the number of ADC bits $b$ for $N_{\rm BS} =128$ BS antennas, $N_{\rm MS} = 8$ MSs, $N_r =16$ selected antennas, and $\rho= 10$ dBm transmit power.} 
	\label{fig:capacity_bit}
\end{figure}

in Fig.~\ref{fig:capacity_power}(a) the QFAS shows higher capacity than FAS, NBS, and random selection cases.
Noting that the initial point of the QMCMC-AS method is the antenna subset from the QFAS, the QMCMC-AS with $(N_{\rm MCMC}= 6, \tau_{\rm stop} = 3)$ provides no capacity increase from the QFAS method. 
Although the QMCMC-AS with $(60, 30)$ shows capacity increase from the QFAS method, it is marginal.
Accordingly, the QFAS method achieves a near optimal performance in terms of capacity with low complexity.
The FAS method offers marginal improvement from the random selection case as it ignores quantization error associated with selected antennas.
The NBS method shows the worst performance in low-resolution ADC systems, which means that selecting the subset of antennas that gives the largest channel gains not only increases the inter-user interference but also increases quantization error.

\begin{figure}[!t]
\centering
$\begin{array}{c c}
{\resizebox{0.48\columnwidth}{!}
{\includegraphics{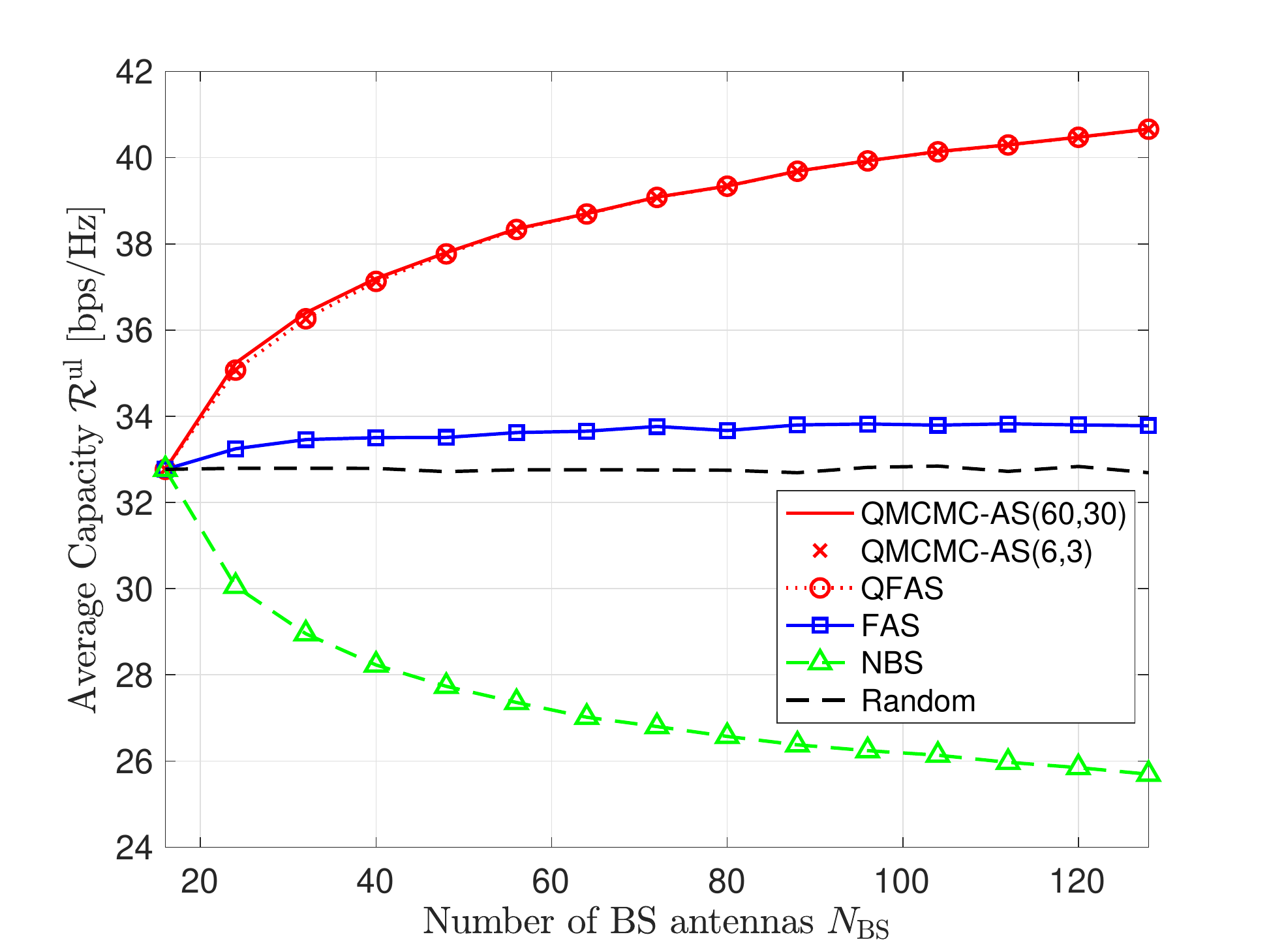}}
} &
{\resizebox{0.48\columnwidth}{!}
{\includegraphics{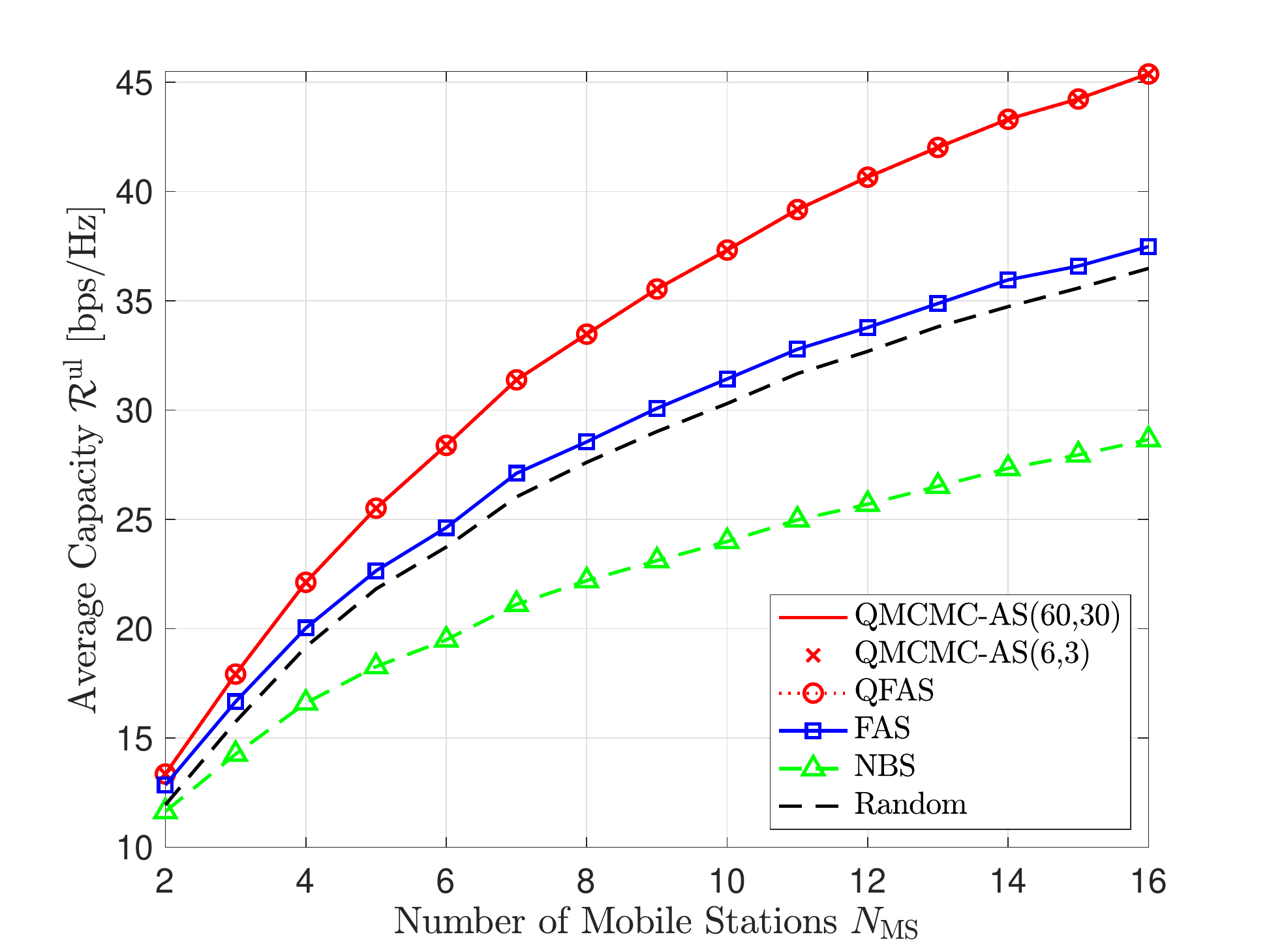}}
}\\
\mbox{\small (a)}& \mbox{\small (b)}
\end{array}$
\caption{
Average capacity $\mathcal{R}^{\rm ul}$ (a) with respect to the number of BS antennas $N_{\rm BS}$ for $N_{\rm MS} = 12$ MSs, $N_r =16$ selected antennas, $\rho= 20$ dBm transmit power, and $b = 3$  ADC bits, and (b) with respect to the number of MSs $N_{\rm MS}$ for $N_{\rm BS} = 128$ BS antennas, $N_r =16$ selected antennas, $\rho= 20$ dBm transmit power, and $b = 3$  ADC bits.} 
\label{fig:capacity_Nbs_Nms}
\end{figure}

With the increased number of receive antennas, selected antennas, and  MSs, the trend of the curves in Fig.~\ref{fig:capacity_power}(b) is similar to Fig.~\ref{fig:capacity_power}(a). 
The QMCMC-AS with $(60,30)$, however, shows no improvement from the QFAS. 
This shows that the QMCMC-AS is not scalable with the number of BS antennas and selected antennas.  
In both Fig.~\ref{fig:capacity_power}(a) and (b), the capacity gap between the QFAS algorithm and the conventional algorithms increases with the transmit power $\rho$ because the quantization error becomes more dominant than the AWGN as the transmit power increases.
In addition, the results in Fig.~\ref{fig:capacity_power} demonstrate that the conventional UL antenna selection approaches are not applicable to the low-resolution ADC receivers.

In Fig.~\ref{fig:capacity_bit}, we note that in the low-resolution ADC regime, the capacity of the QFAS method is higher than the FAS, NBS, and random selection.
This corresponds to the intuition for the proposed method such that considering the quantization error is critical when selecting antennas in low-resolution ADC systems.
The capacity of the QFAS and FAS methods converges as the number of ADC bits $b$ increases, thereby showing that the proposed QFAS method is generalized version of the FAS in terms of quantization precision.
The NBS method performs better than the random selection in high-resolution ADC regime while it still performs worse in the low-resolution ADC regime.
Again, this validates the intuition that the antenna selection approaches for high-resolution ADC systems cannot directly be applied to the low-resolution ADC receivers.

\begin{figure}[!t]\centering
	\includegraphics[scale = 0.43]{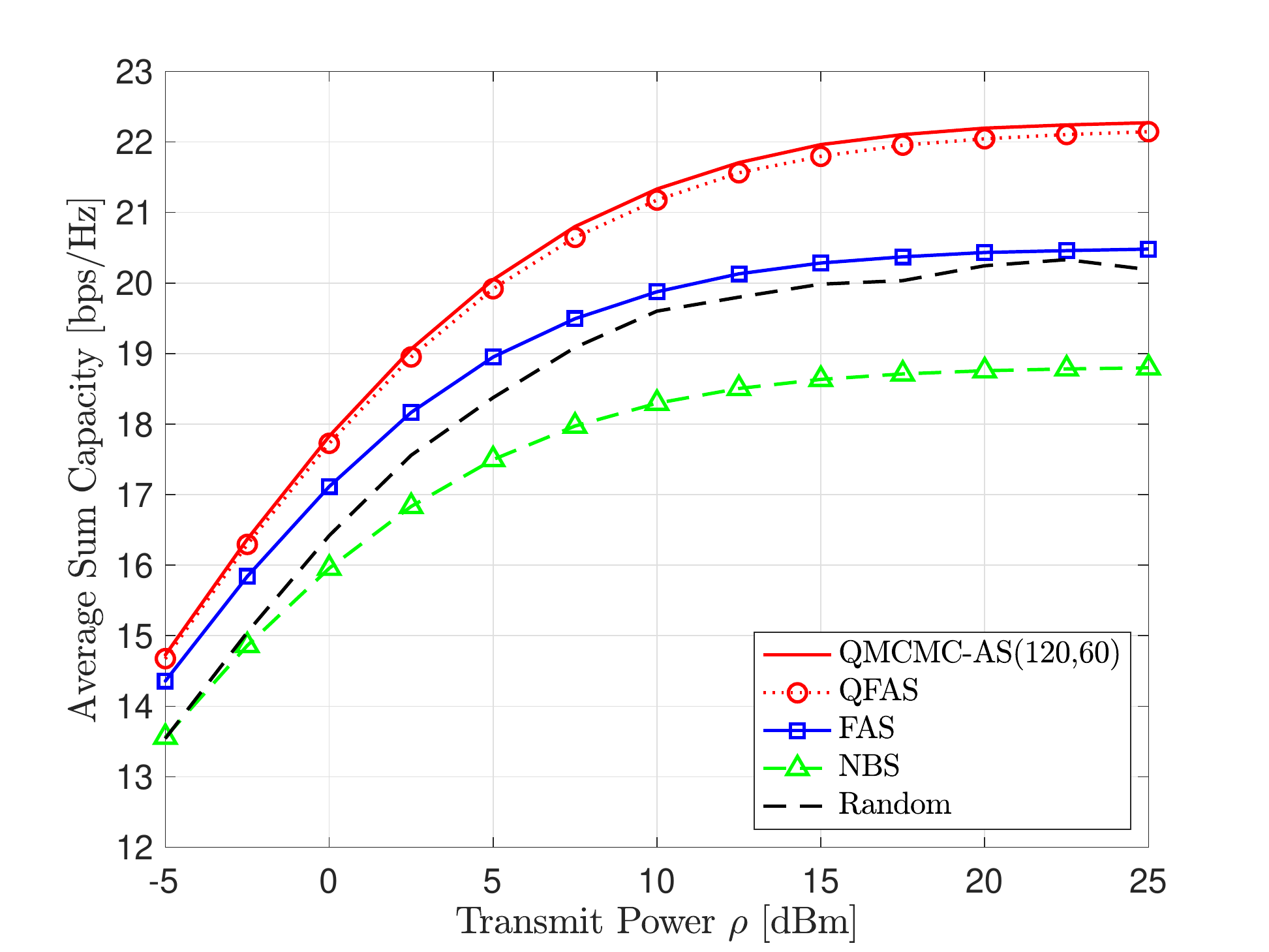}
	\caption{Average sum capacity $\frac{1}{N_{\rm sc}}\sum_{n}\mathcal{R}^{\rm ul}_n$ with respect to transmit power $\rho$ for $N_{\rm BS} = 32$ BS antennas, $N_{\rm MS} = 8$ MSs, $N_r =8$ selected antennas, $b = 3$ quantization bits, and $N_{\rm sc} = 64$ subcarriers with $L =4$-tap channels.} 
	\label{fig:capacity_power_ofdm}
\end{figure}

In Fig.~\ref{fig:capacity_Nbs_Nms}(a), we observe large improvement from the random selection for the QFAS method as $N_{\rm BS}$ increases whereas the FAS and NBS cannot provide such improvement. 
Accordingly, the proposed QFAS method can be effective in the large antenna array systems with low-resolution ADCs by efficiently reducing the number of RF chains.
We note that the capacity with the NBS method even decreases with the number of BS antennas since the increased candidate antenna size worsens the resulting subset of antennas by significantly increasing quantization error and interference.
In Fig.~\ref{fig:capacity_Nbs_Nms}(b), the capacity gap between the QFAS and FAS methods increases with $N_{\rm MS}$, which is desirable in term of maximizing the sum rate.
Overall, performance improvement with the proposed QFAS becomes larger as more users are served and more antennas are deployed for the fixed number of selected antennas (equivalently RF chains), which is desirable for future communication systems that are likely to serve more users with more antennas.


\subsubsection{Wideband OFDM Communications}

we consider UL wideband ODFM communications with $N_{\rm sc} = 64$ subcarriers for channels with $L = 4$ taps.
Similarly to the simulation results for the narrowband system, the proposed QFAS method in Fig.~\ref{fig:capacity_power_ofdm} shows higher capacity than the FAS, NBS, and random selection. 
In addition, the QFAS method almost achieves the capacity of the QMCMC-AS with the increased number of sampling and iterations $(N_{\rm MCMC} =120, \tau_{\rm stop} = 60)$.
Therefore, the QFAS can also achieve near optimal selection performance in wideband OFDM systems while the FAS method shows marginal improvement from the random selection and the NBS method shows the worst performance in low-resolution ADC systems.

In Fig.~\ref{fig:capacity_Nr_ofdm}, we note that the proposed QFAS performs better than the FAS, NBS, and random selection for any size of antenna subset $N_r$.
The QFAS provides saving of about $10$ RF chains on average compared to the FAS and random selection,
Such saving can be considered as large for receivers with the relatively small number RF chains compared to the number of antennas.
Overall, the simulation results demonstrate that the conventional receive antenna selection is not adequate under non-negligible quantization error and that the proposed QFAS can effectively incorporate the quantization error in antenna selection.

\begin{figure}[!t]\centering
	\includegraphics[scale = 0.43]{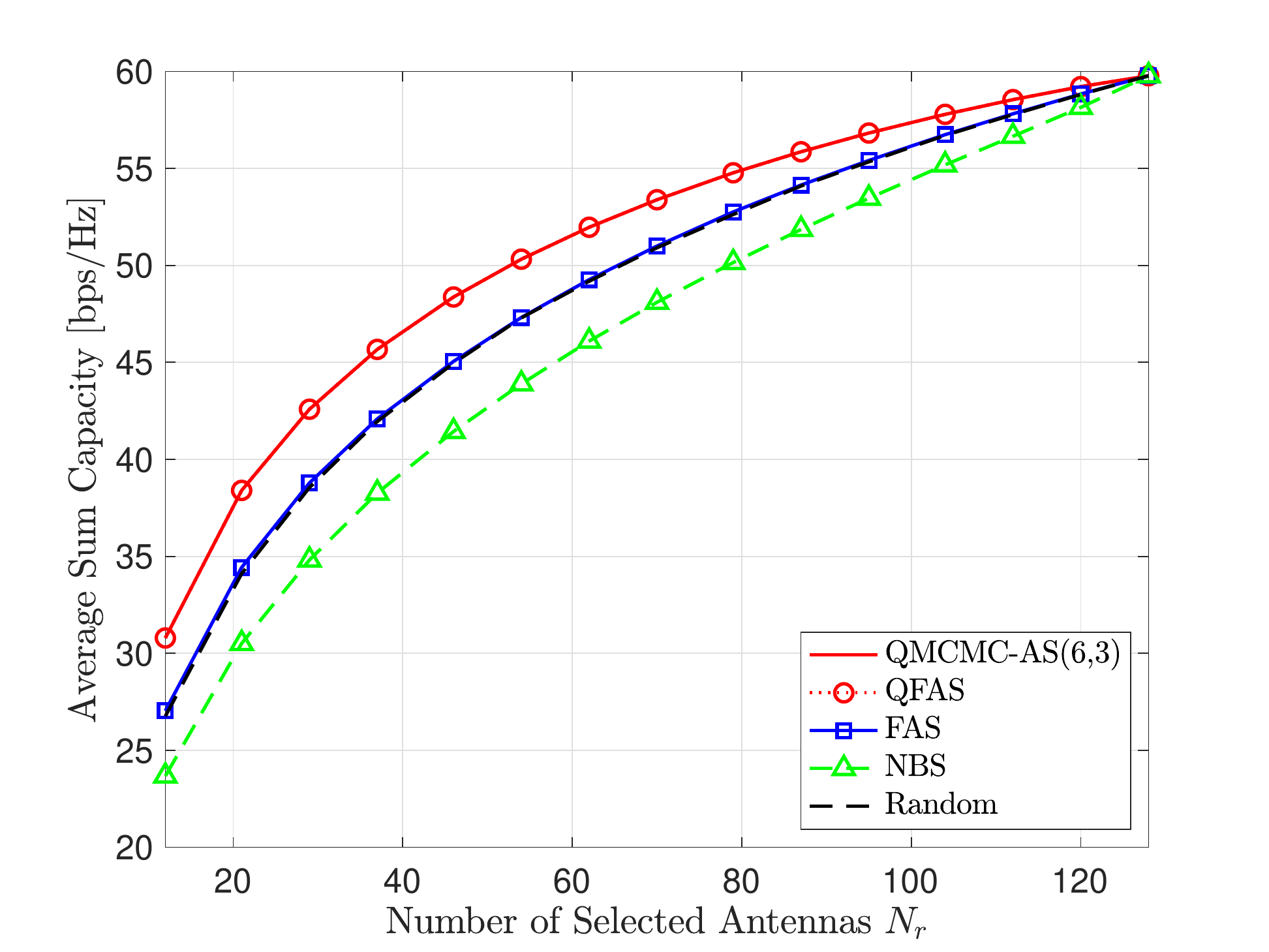}
	\caption{Average sum capacity $\frac{1}{N_{\rm sc}}\sum_{n}\mathcal{R}^{\rm ul}_n$ with respect to the number of selected antennas $N_r$ for $N_{\rm BS} = 128$ BS antennas, $N_{\rm MS} = 12$ MSs, $b = 3$ quantization bits, $N_{\rm sc} = 64$ subcarriers with $L =4$-tab channels, and $\rho = 20$ dBm.} 
	\label{fig:capacity_Nr_ofdm}
\end{figure}

\section{Conclusion}
\label{sec:conclusion}

In this paper, we investigate antenna selection at a BS in low-resolution ADC systems to achieve power-efficient wireless communication systems.
For downlink narrowband and wideband OFDM systems, we showed that the existing state-of-the-art transmit antenna selection techniques can be applicable to the low-resolution ADC systems when the BS employs the ZF precoding with equal power distribution. 
In addition, we proved that it is beneficial to use more antennas in terms of maximizing the sum rate.
Unlike the high-resolution ADC systems, we validated that the transmit antenna selection can achieve a comparable sum rate to the system that uses all antennas by increasing the total transmit power constraint, which allows to reduce the number of RF chains with marginal sum rate loss.
For an uplink narrowband and wideband OFDM systems, we showed that the conventional receive antenna selection criteria are insufficient for the low-resolution ADC systems.
The generalized greedy selection criterion provided that capturing the balance between the channel gain and increase in quantization error is critical when there is non-negligible quantization error at the receiver.
The propose greedy selection algorithm showed that it guarantees $(1-\frac{1}{e})$ of the capacity with an optimal antenna subset.
In simulations, theoretical analyses were validated and the proposed algorithms outperformed the conventional algorithms in capacity, achieving a near optimal performance with low complexity. 

\bibliographystyle{IEEEtran}
\bibliography{BSAntSel.bib}

\begin{thebibliography}{10}
\providecommand{\url}[1]{#1}
\csname url@samestyle\endcsname
\providecommand{\newblock}{\relax}
\providecommand{\bibinfo}[2]{#2}
\providecommand{\BIBentrySTDinterwordspacing}{\spaceskip=0pt\relax}
\providecommand{\BIBentryALTinterwordstretchfactor}{4}
\providecommand{\BIBentryALTinterwordspacing}{\spaceskip=\fontdimen2\font plus
\BIBentryALTinterwordstretchfactor\fontdimen3\font minus
  \fontdimen4\font\relax}
\providecommand{\BIBforeignlanguage}[2]{{%
\expandafter\ifx\csname l@#1\endcsname\relax
\typeout{** WARNING: IEEEtran.bst: No hyphenation pattern has been}%
\typeout{** loaded for the language `#1'. Using the pattern for}%
\typeout{** the default language instead.}%
\else
\language=\csname l@#1\endcsname
\fi
#2}}
\providecommand{\BIBdecl}{\relax}
\BIBdecl

\bibitem{choi2018antenna}
J.~Choi, J.~Sung, B.~L. Evans, and A.~Gatherer, ``{Antenna selection for
  large-scale MIMO systems with low-resolution ADCs},'' in \emph{IEEE Int.
  Conf. on Acoustics, Speech and Signal Process.}, Apr. 2018, pp. 3594--3598.

\bibitem{marzetta2010noncooperative}
T.~L. Marzetta, ``{Noncooperative cellular wireless with unlimited numbers of
  base station antennas},'' \emph{IEEE Trans. on Wireless Commun.}, vol.~9,
  no.~11, p. 3590, Nov. 2010.

\bibitem{larsson2014massive}
E.~G. Larsson, O.~Edfors, F.~Tufvesson, and T.~L. Marzetta, ``{Massive MIMO for
  next generation wireless systems},'' \emph{IEEE Comm. Mag.}, vol.~52, no.~2,
  pp. 186--195, Feb. 2014.

\bibitem{lu2014overview}
L.~Lu, G.~Y. Li, A.~L. Swindlehurst, A.~Ashikhmin, and R.~Zhang, ``{An overview
  of massive MIMO: benefits and challenges},'' \emph{IEEE Journal of Sel.
  Topics in Signal Process.}, vol.~8, no.~5, pp. 742--758, Oct. 2014.

\bibitem{mendez2016hybrid}
R.~M\'endez-Rial, C.~Rusu, N.~Gonz\'alez-Prelcic, A.~Alkhateeb, and R.~W.
  Heath, ``{Hybrid MIMO architectures for millimeter wave Communications: Phase
  shifters or switches?}'' \emph{IEEE Access}, vol.~4, pp. 247--267, Jan. 2016.

\bibitem{walden1999analog}
R.~H. Walden, ``{Analog-to-digital converter survey and analysis},'' \emph{IEEE
  Journal on Sel. Areas in Commun.}, vol.~17, no.~4, pp. 539--550, Apr. 1999.

\bibitem{mezghani2007ultra}
A.~Mezghani and J.~A. Nossek, ``{On ultra-wideband MIMO systems with 1-bit
  quantized outputs: Performance analysis and input optimization},'' in
  \emph{IEEE Int. Symposium on Inform. Theory}, Jul. 2007, pp. 1286--1289.

\bibitem{mo2015capacity}
J.~Mo and R.~W. Heath, ``{Capacity analysis of one-bit quantized MIMO systems
  with transmitter channel state information},'' \emph{IEEE Trans. on Signal
  Process.}, vol.~63, no.~20, pp. 5498--5512, Oct. 2015.

\bibitem{chen2005analysis}
Z.~{Chen}, J.~{Yuan}, and B.~{Vucetic}, ``{Analysis of transmit antenna
  selection/maximal-ratio combining in Rayleigh fading channels},'' \emph{IEEE
  Trans. on Veh. Technol.}, vol.~54, no.~4, pp. 1312--1321, Jul. 2005.

\bibitem{sanayei2007capacity}
S.~{Sanayei} and A.~{Nosratinia}, ``{Capacity of MIMO Channels With Antenna
  Selection},'' \emph{IEEE Trans. on Inform. Theory}, vol.~53, no.~11, pp.
  4356--4362, Nov. 2007.

\bibitem{gao2013antenna}
X.~Gao, O.~Edfors, J.~Liu, and F.~Tufvesson, ``{Antenna selection in measured
  massive MIMO channels using convex optimization},'' in \emph{IEEE Global
  Commun. Conf. Workshops}, Dec. 2013, pp. 129--134.

\bibitem{khademi2014convex}
S.~{Khademi}, E.~{DeCorte}, G.~{Leus}, and A.~{van der Veen}, ``{Convex
  optimization for joint zero-forcing and antenna selection in multiuser MISO
  systems},'' in \emph{IEEE Int. Workshop on Signal Process. Adv. in Wireless
  Commun.}, Jun. 2014, pp. 30--34.

\bibitem{zhang2008performance}
X.~Zhang, Z.~Lv, and W.~Wang, ``{Performance analysis of multiuser diversity in
  MIMO systems with antenna selection},'' \emph{IEEE Trans. on Wireless
  Commun.}, vol.~7, no.~1, pp. 15--21, Jan. 2008.

\bibitem{amadori2017large}
P.~V. {Amadori} and C.~{Masouros}, ``{Large Scale Antenna Selection and
  Precoding for Interference Exploitation},'' \emph{IEEE Trans. on Commun.},
  vol.~65, no.~10, pp. 4529--4542, Oct. 2017.

\bibitem{liu2017energy}
Z.~{Liu}, W.~{Du}, and D.~{Sun}, ``{Energy and Spectral Efficiency Tradeoff for
  Massive MIMO Systems With Transmit Antenna Selection},'' \emph{IEEE Trans. on
  Veh. Technol.}, vol.~66, no.~5, pp. 4453--4457, May 2017.

\bibitem{yang2016transmit}
P.~{Yang}, Y.~{Xiao}, Y.~L. {Guan}, S.~{Li}, and L.~{Hanzo}, ``{Transmit
  Antenna Selection for Multiple-Input Multiple-Output Spatial Modulation
  Systems},'' \emph{IEEE Trans. on Commun.}, vol.~64, no.~5, pp. 2035--2048,
  May 2016.

\bibitem{gorokhov2003receive}
A.~Gorokhov, D.~A. Gore, and A.~J. Paulraj, ``{Receive antenna selection for
  MIMO spatial multiplexing: theory and algorithms},'' \emph{IEEE Trans. on
  Signal Process.}, vol.~51, no.~11, pp. 2796--2807, Dec. 2003.

\bibitem{molisch2005cap}
A.~F. {Molisch}, M.~Z. {Win}, , and J.~H. {Winters}, ``{Capacity of MIMO
  systems with antenna selection},'' \emph{IEEE Trans. on Wireless Commun.},
  vol.~4, no.~4, pp. 1759--1772, July 2005.

\bibitem{dua2006receive}
A.~Dua, K.~Medepalli, and A.~J. Paulraj, ``{Receive antenna selection in MIMO
  systems using convex optimization},'' \emph{IEEE Trans. on Wireless Commun.},
  vol.~5, no.~9, pp. 2353--2357, Sep. 2006.

\bibitem{vaze2012submodularity}
R.~{Vaze} and H.~{Ganapathy}, ``{Sub-Modularity and Antenna Selection in MIMO
  Systems},'' \emph{IEEE Commun. Lett.}, vol.~16, no.~9, pp. 1446--1449, Sep.
  2012.

\bibitem{liu2009low}
Y.~Liu, Y.~Zhang, C.~Ji, W.~Q. Malik, and D.~J. Edwards, ``{A low-complexity
  receive-antenna-selection algorithm for MIMO--OFDM wireless systems},''
  \emph{IEEE Trans. on Veh. Technol.}, vol.~58, no.~6, pp. 2793--2802, Dec.
  2009.

\bibitem{nara2009antenna}
A.~B. {Narasimhamurthy} and C.~{Tepedelenlioglu}, ``{Antenna Selection for
  MIMO-OFDM Systems With Channel Estimation Error},'' \emph{IEEE Trans. on Veh.
  Technol.}, vol.~58, no.~5, pp. 2269--2278, Jun 2009.

\bibitem{zhang2009receive}
Y.~{Zhang}, C.~{Ji}, W.~Q. {Malik}, D.~C. {O'Brien}, and D.~J. {Edwards},
  ``{Receive antenna selection for MIMO systems over correlated fading
  channels},'' \emph{IEEE Trans. on Wireless Commun.}, vol.~8, no.~9, pp.
  4393--4399, Sep. 2009.

\bibitem{dai2006optimal}
L.~{Dai}, S.~{Sfar}, and K.~B. {Letaief}, ``{Optimal antenna selection based on
  capacity maximization for MIMO systems in correlated channels},'' \emph{IEEE
  Trans. on Commun.}, vol.~54, no.~3, pp. 563--573, Mar. 2006.

\bibitem{amadori2015low}
P.~V. Amadori and C.~Masouros, ``{Low RF-complexity millimeter-wave
  beamspace-MIMO systems by beam selection},'' \emph{IEEE Trans. on Commun.},
  vol.~63, no.~6, pp. 2212--2223, May 2015.

\bibitem{li2019joint}
H.~{Li}, Q.~{Liu}, Z.~{Wang}, and M.~{Li}, ``{Joint Antenna Selection and
  Analog Precoder Design With Low-Resolution Phase Shifters},'' \emph{IEEE
  Trans. on Veh. Technol.}, vol.~68, no.~1, pp. 967--971, Jan 2019.

\bibitem{torabi2008antenna}
M.~Torabi, ``{Antenna selection for MIMO-OFDM systems},'' \emph{Elsevier Signal
  Process.}, vol.~88, no.~10, pp. 2431--2441, 2008.

\bibitem{le2016antenna}
N.~P. {Le}, F.~{Safaei}, and L.~C. {Tran}, ``{Antenna Selection Strategies for
  MIMO-OFDM Wireless Systems: An Energy Efficiency Perspective},'' \emph{IEEE
  Trans. on Veh. Technol.}, vol.~65, no.~4, pp. 2048--2062, April 2016.

\bibitem{chen2019joint}
J.-C. Chen, ``{Joint Antenna Selection and User Scheduling for Massive
  Multiuser MIMO Systems With Low-Resolution ADCs},'' \emph{IEEE Trans. on Veh.
  Technol.}, vol.~68, no.~1, pp. 1019--1024, Nov. 2019.

\bibitem{choi2019analysis}
J.~Choi and B.~L. Evans, ``{Analysis of Ergodic Rate for Transmit Antenna
  Selection in Low-Resolution ADC Systems},'' \emph{IEEE Trans. on Veh.
  Technol.}, vol.~68, no.~1, pp. 952--956, Oct. 2019.

\bibitem{fletcher2007robust}
A.~K. Fletcher, S.~Rangan, V.~K. Goyal, and K.~Ramchandran, ``{Robust
  predictive quantization: Analysis and design via convex optimization},''
  \emph{IEEE Journal of Sel. Topics in Signal Process.}, vol.~1, no.~4, pp.
  618--632, Dec. 2007.

\bibitem{orhan2015low}
O.~Orhan, E.~Erkip, and S.~Rangan, ``{Low power analog-to-digital conversion in
  millimeter wave systems: Impact of resolution and bandwidth on
  performance},'' in \emph{IEEE Inform. Theory and App. Work.}, Feb. 2015, pp.
  191--198.

\bibitem{gersho2012vector}
A.~Gersho and R.~M. Gray, \emph{{Vector quantization and signal
  compression}}.\hskip 1em plus 0.5em minus 0.4em\relax Springer 2012
  (originally published 1992).

\bibitem{fan2015uplink}
L.~Fan, S.~Jin, C.-K. Wen, and H.~Zhang, ``{Uplink achievable rate for massive
  MIMO systems with low-resolution ADC},'' \emph{IEEE Commun. Lett.}, vol.~19,
  no.~12, pp. 2186--2189, Oct. 2015.

\bibitem{lin2012performance}
P.-H. Lin and S.-H. Tsai, ``{Performance analysis and algorithm designs for
  transmit antenna selection in linearly precoded multiuser MIMO systems},''
  \emph{IEEE Trans. on Veh. Technol.}, vol.~61, no.~4, pp. 1698--1708, Mar.
  2012.

\bibitem{gharavi2004fast}
M.~Gharavi-Alkhansari and A.~B. Gershman, ``{Fast antenna subset selection in
  MIMO systems},'' \emph{IEEE Trans. on Signal Process.}, vol.~52, no.~2, pp.
  339--347, Feb. 2004.

\bibitem{nemhauser1978analysis}
G.~L. Nemhauser, L.~A. Wolsey, and M.~L. Fisher, ``{An analysis of
  approximations for maximizing submodular set functions--I},'' \emph{Math.
  Programming}, vol.~14, no.~1, pp. 265--294, 1978.

\bibitem{liu2008monte}
J.~S. Liu, \emph{{Monte Carlo strategies in scientific computing}}.\hskip 1em
  plus 0.5em minus 0.4em\relax Springer Science \& Business Media, 2008.

\bibitem{harold1997stochastic}
J.~Harold, G.~Kushner, and Y.~George, ``{Stochastic Approximation Algorithms
  and Applications},'' 1997.

\bibitem{prasad2019optimizing}
N.~Prasad, X.-F. Qi, and A.~Gatherer, ``{Optimizing beams and bits: A novel
  approach for massive MIMO base station design},'' in \emph{IEEE Int. Conf. on
  Computing, Networking and Commun.}, Apr. 2019, pp. 970--976.

\bibitem{erceg1999empirically}
V.~Erceg, L.~J. Greenstein, S.~Y. Tjandra, S.~R. Parkoff, A.~Gupta, B.~Kulic,
  A.~A. Julius, and R.~Bianchi, ``{An empirically based path loss model for
  wireless channels in suburban environments},'' \emph{IEEE Journal on Sel.
  Areas in Commun.}, vol.~17, no.~7, pp. 1205--1211, Jul. 1999.

\end{thebibliography}
\end{document}